\pdfoutput=1
\RequirePackage{fix-cm}
\documentclass[11pt]{article}
\usepackage[margin=1in]{geometry}
\usepackage{mathtools,amssymb,amsthm}
\usepackage[varbb]{newpxmath}
\usepackage{microtype}
\usepackage{graphicx}
\usepackage{capt-of,float,needspace,multicol,booktabs,array}
\usepackage[algo2e,ruled]{algorithm2e}
\usepackage[nolist]{acronym}
\newacro{FTRL}{Follow-the-Regularized-Leader}
\usepackage[dvipsnames,svgnames]{xcolor}
\usepackage{tikz}
\usetikzlibrary{arrows.meta,positioning}
\usepackage{tcolorbox}
\tcbuselibrary{skins,breakable}
\usepackage{pgfplots}
\usepgfplotslibrary{groupplots}
\pgfplotsset{compat=1.18}
\usepackage{hyperref}
\definecolor{blueGrotto}{HTML}{059DC0}
\definecolor{royalBlue}{HTML}{057DCD}
\definecolor{navyBlue}{HTML}{0B579C}
\hypersetup{
  colorlinks=true,
  urlcolor=blueGrotto,
  linkcolor=royalBlue,
  citecolor=navyBlue
}
\newtagform{referenceblue}[\textcolor{royalBlue}]{\textcolor{black}{(}}{\textcolor{black}{)}}
\usetagform{referenceblue}

\usepackage[
  backref=true,
  backend=biber,
  natbib=true,
  style=alphabetic,
  sorting=alphabeticlabel,
  sortcites=true,
  minbibnames=3,
  maxbibnames=999,
  mincitenames=4,
  maxcitenames=4,
  minalphanames=4,
  maxalphanames=4
]{biblatex}
\DeclareSortingTemplate{alphabeticlabel}{
  \sort[final]{\field{labelalpha}}
  \sort{\field{year}}
  \sort{\field{title}}
}
\AtBeginRefsection{\GenRefcontextData{sorting=ynt}}
\AtEveryCite{\localrefcontext[sorting=ynt]}
\usepackage{xurl}
\theoremstyle{plain}
\newtheorem{theorem}{Theorem}
\newtheorem{lemma}{Lemma}
\newtheorem{proposition}{Proposition}
\newtheorem{corollary}{Corollary}

\newtheorem{maintheorem}[theorem]{Theorem}
\newtheorem*{gapfact}{Fact (Gap Independent Set; Dinur--Safra)}
\theoremstyle{definition}
\newtheorem{definition}{Definition}
\newtheorem{assumption}{Assumption}

\theoremstyle{remark}
\newtheorem{remark}{Remark}
\newenvironment{proofof}[1]{\begin{proof}[Proof of #1]}{\end{proof}}

\definecolor{brightblue}{RGB}{0,90,210}
\tcolorboxenvironment{maintheorem}{
  enhanced, breakable, lines before break=3, frame hidden, sharp corners,
  colback=brightblue!4,
  borderline west={1.2pt}{0pt}{brightblue},
  left=4pt, right=4pt, top=4pt, bottom=4pt, boxsep=0pt,
  oversize,
  before skip=7pt plus 2pt, after skip=7pt plus 2pt
}

\title{The Complexity of Strategizing Against No-Regret Learners:\\
Pure When It Is Easy, Hard Even When Mixed}
\date{}
\author{
Yiheng Su\\
University of Wisconsin--Madison\\
\texttt{ysu24@cs.wisc.edu}
\and
Emmanouil-Vasileios Vlatakis-Gkaragkounis\\
University of Wisconsin--Madison\\
\texttt{vlatakis@wisc.edu}
}
\hypersetup{
  pdftitle={The Complexity of Strategizing Against No-Regret Learners: Pure When It Is Easy, Hard Even When Mixed},
  pdfauthor={Yiheng Su and Emmanouil-Vasileios Vlatakis-Gkaragkounis}
}

\begin{document}
\pagenumbering{roman}
\maketitle
\bigskip
\bigskip
\thispagestyle{empty}
% !TEX root = ./main.tex
\begin{abstract}
We study the computational complexity of clairvoyant
strategy optimization against a known online learner
in finite-horizon repeated games. In general-sum games
with bounded payoffs, unless $\mathsf{P}=\mathsf{NP}$, no polynomial-time
optimizer can achieve a reward within $cT$
of the reward of the optimal pure-action
sequence for an absolute constant $c>0$, even
when allowed to output explicitly represented mixed
strategies at every round. Our hardness result
strengthens prior work along two complementary axes:
it extends to a broader class of
learning dynamics, including widely used mean-based algorithms—answering
an open question of \citet{assos2025intractability} for learners
beyond multiplicative weights—and it permits mixed outputs
while retaining the weaker pure-sequence benchmark. Thus,
neither mixed strategies nor algorithmic specificity removes the
computational obstruction in general-sum games. In sharp
contrast, zero-sum games exhibit the reverse phenomenon:
deterministic pure actions suffice for efficiently recovering
the leading surplus guarantees previously obtained
through mixed play against \ac{FTRL} opponents.

Thus, the two regimes form a striking
mirror image: mixed play cannot overcome general-sum
hardness even against a pure benchmark, whereas
pure play retains the full surplus guarantees
of mixed strategies in zero-sum games.
\end{abstract}
\acresetall

\clearpage
\tableofcontents
\thispagestyle{empty}
\clearpage
\pagenumbering{arabic}

% !TEX root = ../main.tex
\section{Introduction}
\label{sec:introduction}
Learning in games is usually studied in
the \emph{forward direction}: fix the environment and
ask whether the learner converges or incurs
little regret. We ask how to shape
that behavior. An \emph{optimizer} knows both players'
payoffs and the learner's update rule, but
today's actions affect tomorrow's responses. Sacrificing reward
now may yield greater gains later.
\par\begingroup
\setlength{\topsep}{5pt}
\setlength{\partopsep}{0pt}
\begin{center}
\emph{Can
the optimizer efficiently find a near-optimal sequence
of actions?}
\end{center}
\endgroup

A no-regret guarantee does not settle this
question. It compares the learner with a
single action chosen in hindsight, against the
same payoff sequence. It says little about
choosing that sequence strategically. Against rules such
as multiplicative weights (Hedge/MWU), steering can outperform
the one-shot Stackelberg commitment benchmark by an
amount proportional to the interaction's duration \citep{braverman2017sellingnoregretbuyer,deng2019strategizing}.
Yet computing near-optimal steering against Hedge/MWU is
hard, even to within a fixed loss
per round \citep{assos2025intractability}.

Could more freedom help? A \emph{pure action}
selects one available action; a \emph{mixed action}
assigns a probability distribution over them. Mixing
lets the optimizer combine several directions of
influence in each round. In one-shot commitment
problems, this freedom supports efficient optimization by
linear programming \citep{conitzer2006commitment}.
\par\begingroup
\setlength{\topsep}{5pt}
\setlength{\partopsep}{0pt}
\begin{center}
\begin{minipage}{0.82\linewidth}
\centering
\emph{Does allowing mixed actions make near-optimal steering
of known learning dynamics computationally tractable?}
\end{minipage}
\end{center}
\endgroup
Our work gives a sweeping answer. \emph{In
general-sum games, mixed play does not remove
the hardness. In zero-sum games, pure play
already suffices.} Our hardness result and constructive
guarantees explain the contrast in our title:
pure when it is easy, hard even
when mixed.

\subsection{Our contributions}
The optimizer plans over a finite horizon
against a known deterministic learning rule. Its
task is to compute a high-reward sequence
in time polynomial in the horizon and
the game description.

\paragraph{General-sum games: hard even when mixed.}
Our hardness result strengthens prior work along
two axes: it gives the optimizer more
freedom and covers a broader class of
learners. First, the optimizer may use mixed
actions at every round, while the benchmark
remains the best pure sequence. Even then,
unless $\mathsf{P}=\mathsf{NP}$, no polynomial-time optimizer can guarantee
a reward within a fixed loss per
round of this benchmark (Theorem~\ref{thm:main}). With bounded
payoffs, the gap grows linearly with the
horizon. This strengthens the Hedge/MWU result of
\citet{assos2025intractability}, which allows only pure optimizer
outputs. All four pure/mixed output--benchmark combinations are
hard (Corollary~\ref{cor:four-tasks}).

Second, the reduction uses only a quantitative
\emph{mean-based condition} \citep{braverman2017sellingnoregretbuyer}: an action whose cumulative
payoff falls sufficiently far behind another receives
little probability. Under the theorem's rate requirement,
hardness holds separately for every fixed deterministic
learner satisfying this condition, even when the
optimizer is tailored to its rule. Beyond
Hedge/MWU, this covers Euclidean, Tsallis, and log-barrier
regularization under the learning-rate conditions in Appendix~\ref{app:learner-classes},
answering the question of \citet{assos2025intractability} about other
learners.

\paragraph{Zero-sum games: pure play already gets close.}
When one player's gain is the other's
loss, the game value supplies a baseline:
the reward per round the optimizer can
secure against any opponent. At each round,
a pure best response to the learner's
current distribution earns at least this value.
The learner's regret bound caps how far
any trajectory can rise above it. This
pure strategy is therefore within the learner's
regret bound of the optimal mixed trajectory,
and hence of the pure optimum (Proposition~\ref{prop:app-pure-best-response}).
Against any no-regret learner, its average shortfall
vanishes. Computing the sequence is efficient whenever
the learner's responses can be computed efficiently.

\paragraph{Beyond the game value.}
The finer question is whether pure play
also retains the \emph{surplus} above the game-value
baseline previously obtained through mixed play. Answering
it requires looking inside the learner's update,
so we specialize to separable \ac{FTRL}, including
Hedge/MWU and Euclidean and Tsallis regularization. Following
\citet{su2026nocoin}, we distinguish two mechanisms. \emph{Passive surplus}
accrues while the learner abandons actions that
are suboptimal against a max-min strategy. \emph{Active
surplus} comes from its delayed response to
changing play and can remain available when
the passive mechanism vanishes, as in Rock--Paper--Scissors.
For a fixed game and regularizer, their
characteristic scales are $1/\eta_T$ and $\eta_TT$, respectively,
where $\eta_T$ is the step size. We
recover both leading rates with deterministic pure
sequences, up to rounding losses, under the
respective assumptions of Theorems~\ref{thm:app-pure-passive} and~\ref{thm:app-pure-active}.

\paragraph{Technical challenges: extraction and scheduling.}
For hardness, mixing lets the optimizer spread
its effort among many actions and interleave
them throughout the interaction. Our reduction must
turn any sufficiently rewarding trajectory into a
solution of a hard combinatorial problem. We
start from gap Independent Set \citep{dinursafra2005vertexcover}: vertices
become choices that must be activated and
then harvested for reward. Additional learner actions
act as alarms, limiting reward when incompatible
vertices are activated. The key step extracts
an independent set directly from cumulative optimizer
action counts, whether integral or fractional. Grouping
rounds into intervals of comparable setup counts
makes this possible. The mean-based condition then
bounds reward by the extracted set's size,
without simulating the learner; the source gap
becomes a linear gap in reward.

For zero-sum surplus, the challenge is to
turn a mixture into a schedule. Blocks
and alternations with the same action counts
can produce different histories. The passive lower
bound depends only on the learner's terminal
score, so rounding the target counts suffices:
any ordering retains the guarantee, with a
loss depending only on the mixture's support
size. Active surplus requires control throughout the
trajectory. Cumulative rounding \citep{golz2025apportionment} keeps each action's count
within one unit of its target at
every prefix, while pure updates expose payoff
differences that the mixture averages away. A
derandomization step selects a deterministic sequence retaining
the active-surplus bound.

\subsection{Related work}
Table~\ref{tab:results-comparison} compares these guarantees with prior work.
Bold entries highlight the broader learner classes
and stronger pure/mixed guarantees. Appendix~\ref{app:extended-related-work} discusses the
wider literature.
% !TEX root = ../main.tex
\begin{table}[!htbp]
\centering
\small
\setlength{\tabcolsep}{2.5pt}
\setlength{\aboverulesep}{1.5pt}
\setlength{\belowrulesep}{1.5pt}
\setlength{\abovecaptionskip}{5pt}
\renewcommand{\arraystretch}{1.0}
\newcommand{\resultyes}{\textcolor{green!45!black}{$\checkmark$}}
\newcommand{\resultno}{\textcolor{red!75!black}{$\times$}}
\begin{tabular}{@{}
  >{\raggedright\arraybackslash}p{.28\linewidth}
  >{\raggedright\arraybackslash}p{.17\linewidth}
  >{\raggedright\arraybackslash}p{.085\linewidth}
  >{\raggedright\arraybackslash}p{.155\linewidth}
  >{\raggedright\arraybackslash}p{\dimexpr.31\linewidth-8\tabcolsep\relax}@{}}
\toprule
\textbf{Work} & \textbf{Learner} & \textbf{Output}
  & \textbf{Benchmark} & \textbf{Guarantee} \\
\midrule
\multicolumn{5}{@{}l}{\textit{General-sum games}} \\
\citet{assos2024maximizing}
  & Fictitious play & Pure & Pure optimum
  & \resultno\ Constant additive gap \\
\citet{assos2025intractability}
  & Hedge/MWU & Pure & Pure optimum
  & \resultno\ Linear additive gap \\
\textbf{This work} (Thm.~\ref{thm:main})
  & \textbf{Mean-based}, including Hedge & \textbf{Mixed} & Pure optimum (hence mixed)
  & \resultno\ Linear additive gap \\
\midrule
\multicolumn{5}{@{}l}{\textit{Zero-sum games}} \\
\citet{assos2024maximizing}
  & Hedge/MWU & Mixed & Mixed optimum
  & \resultyes\ Additive $O(\eta T)$ error \\
\textbf{This work} (Prop.~\ref{prop:app-pure-best-response})
  & \textbf{No-regret} & \textbf{Pure} & Mixed optimum (hence pure)
  & \resultyes\ Additive regret bound \\
\citet{su2026nocoin}
  & Separable \acs{FTRL} & Mixed & Game value
  & \resultyes\ Passive/active surplus \\
\textbf{This work} (Thms.~\ref{thm:app-pure-passive}--\ref{thm:app-pure-active})
  & Separable \acs{FTRL} & \textbf{Pure} & Game value
  & \resultyes\ Passive/active surplus \\
\bottomrule
\end{tabular}

\caption{\small Discrete-time guarantees under the stated assumptions.
$\eta$ is the step size. \resultyes\ constructive
result; \resultno\ hardness unless $\mathsf{P}=\mathsf{NP}$.
Optima cover full sequences; errors/gaps concern total
reward. Since $R_{\mathrm{pure}}^\star\le R_{\mathrm{mix}}^\star$, hardness against the
pure optimum extends to the mixed optimum
(Corollary~\ref{cor:four-tasks}); a guarantee against the mixed optimum
also covers the pure optimum (Proposition~\ref{prop:app-pure-best-response}).
Passive/active surplus guarantees are separate; ours retain
their rates up to rounding losses. Mixed
output includes pure actions.}
\label{tab:results-comparison}

\end{table}

\paragraph{Organization.}
Section~\ref{sec:model-result} presents the model. Sections~\ref{sec:hardness} and~\ref{app:zero-sum} give
the general-sum hardness and zero-sum guarantees; Section~\ref{sec:discussion}
discusses their scope. The appendix begins with
a proof roadmap. Appendices~\ref{app:extended-related-work}--\ref{app:open-questions} expand the related work
and future directions.

% !TEX root = ../main.tex
\section{Preliminaries}
\label{sec:model-result}

\subsection{Repeated-game model and benchmarks}

\begingroup
\setlength{\abovedisplayskip}{3pt}
\setlength{\belowdisplayskip}{3pt}
\setlength{\abovedisplayshortskip}{0pt}
\setlength{\belowdisplayshortskip}{3pt}

Our problem follows the classical paradigm of
two-player finite general-sum repeated games. The optimizer
and learner have action sets $\mathcal{A}$
and $\mathcal{B}$, of sizes $n:=|\mathcal{A}|$ and
$m:=|\mathcal{B}|$. Their respective payoff matrices are
$A,B\in\mathbb R^{n\times m}$. For a finite action
set $\mathcal C$, we write $\Delta(\mathcal C)$
for its probability simplex. A \emph{mixed action}
is a distribution in $\Delta(\mathcal C)$; a
\emph{pure action} selects a single $c\in\mathcal C$,
identified with the point mass $e_c$.
At each round
$t$, the optimizer chooses a mixed
action $x_t\in\Delta(\mathcal{A})$, while the learner plays
\[
  y_t=\mathsf{L}_T
  \bigl(t,B^{\mathsf T}x_0,\ldots,B^{\mathsf T}x_{t-1}\bigr)
  \in\Delta(\mathcal{B})
\]
according to a known deterministic full-information learning
rule $\mathsf{L}_T$. The optimizer receives reward $x_t^{\mathsf
T}Ay_t$ and seeks to maximize its
cumulative reward over the horizon $T$.

To track the learner's accumulated payoffs, we
use $S_t:=\sum_{s<t}B^{\mathsf T}x_s$ before round $t$, with
$S_0=0$. We collect mixed trajectories in $\mathcal
X_T:=\Delta(\mathcal{A})^T$ and pure trajectories in $\Pi_T:=\{e_a:a\in\mathcal{A}\}^T\subseteq\mathcal X_T$,
where the powers denote Cartesian products. A
trajectory $x=(x_0,\ldots,x_{T-1})$ earns
\begin{equation*}
  R(x,T):=\sum_{t=0}^{T-1}x_t^{\mathsf T}Ay_t,
  \tag{Cumulative Reward}\label{eq:trajectory-reward}
\end{equation*}
where each $y_t$ is the learner's response
to the preceding actions in $x$. We
use $\pi$ when emphasizing pure play, so
that $\pi_t=e_{a_t}$ for some $a_t\in\mathcal{A}$. The
resulting pure and mixed benchmarks are
\[
  R_{\mathrm{pure}}^\star(T):=\max_{\pi\in\Pi_T}R(\pi,T),
  \qquad
  R_{\mathrm{mix}}^\star(T):=\sup_{x\in\mathcal X_T}R(x,T).
\]
\endgroup
The pure maximum exists because $\Pi_T$ is
finite. Although $\mathcal X_T$ is compact, a
general learning rule need not induce a
continuous reward function, so the mixed supremum
need not be attained. For continuous
response rules, including the \ac{FTRL} model
below, it is a maximum. The
inclusion $\Pi_T\subseteq\mathcal X_T$ gives $R_{\mathrm{mix}}^\star(T)\ge R_{\mathrm{pure}}^\star(T)$.
The learner's \emph{external regret} is the difference
between the payoff of its best fixed
action in hindsight and its actual cumulative
payoff. We use $\operatorname{Regret}_{\mathsf L}(T)\ge0$ for a
uniform bound on this difference:
\[
  \max_{b\in\mathcal B}\sum_{t<T}x_t^{\mathsf T}B(e_b-y_t)
  \le \operatorname{Regret}_{\mathsf L}(T)
  \qquad\text{for every optimizer trajectory }x.
\]
We call the learner \emph{no-regret} when this
bound is $o(T)$. The fixed-action comparator is
evaluated on the same sequence of payoff vectors.

\begin{remark}[Trajectory benchmarks]
\label{rem:benchmarks}
Both benchmarks optimize an entire sequence, allowing
the optimizer's action to change at every
round and evaluating each sequence against the
learner's responses it induces. Observe that this
offline planning objective differs from the fixed-action
comparator in the learner's regret bound.
A regret guarantee alone does not
control the optimizer's trajectory benchmarks in
general-sum games. In zero-sum games, however,
it bounds both benchmarks through the
game value, as shown in Section~\ref{app:zero-sum}.
\end{remark}

\subsection{Learner models}
\label{sec:learner-models}
Our general-sum hardness result applies to every
fixed deterministic learner satisfying the quantitative \emph{mean-based
condition} below, at the rate in Theorem~\ref{thm:main}.
This condition alone addresses the question of
\citet[Section~5]{assos2025intractability} about hardness beyond multiplicative weights.

\paragraph{Mean-based learners.}

Mean-based learning, introduced by \citet{braverman2017sellingnoregretbuyer}, formalizes the
principle that an action whose cumulative payoff
falls sufficiently far behind another action receives
little probability. The parameter $\gamma_T$ controls both
the score gap $\gamma_TT$ and the
remaining probability on an inferior action.

\begin{definition}
A learner is $\gamma_T$-\emph{mean-based} if, at every
round and for every payoff history,
$y_t(b)\le\gamma_T$ whenever $S_t(b')>S_t(b)+\gamma_TT$, for all $b,b'\in\mathcal{B}$.
\end{definition}

Beyond this condition, we impose no structural
restriction on the learning rule: it may
depend on the horizon, the current time,
the ordered payoff history, and arbitrary
fixed internal memory. Thus the hardness
proof needs neither a specific update formula
nor a separate regret assumption.

The basic zero-sum approximation bound uses only
a regret guarantee. Recovering the finer
surplus gains, however, calls for a closer
look at the learner's update rule. Our
analysis therefore uses the regularizer explicitly and
specializes these surplus results to separable \ac{FTRL}.

\paragraph{Separable \acs{FTRL}.}
\label{sec:ftrl-preliminaries}
The geometry of the regularizer determines the
learner's response, as in \citet{giannou2021regularizedrates} and the
strategic-surplus analysis of \citet{su2026nocoin}. To make this
structure precise, we use the following convention.
\begin{definition}[Separable regularizer]
\label{def:separable-regularizer}
A \emph{separable regularizer} is a finite, continuous, strictly
convex function $h:\Delta(\mathcal{B})\to\mathbb R$ of the form
$h(y)=\sum_{b\in\mathcal{B}}\theta(y_b)$, where the kernel
$\theta\in C([0,1])\cap C^2((0,1))$ satisfies
$\theta''(u)>0$ for every $u\in(0,1)$.
\end{definition}

Given a separable regularizer $h$, the learner balances linear payoff
against regularization through the response map
\begin{equation*}
  Q_h(z)
  :=\operatorname*{arg\,max}_{y\in\Delta(\mathcal{B})}
  \bigl\{\langle z,y\rangle-h(y)\bigr\}.
  \tag{\acs{FTRL} response}
  \label{eq:app-ftrl-response}
\end{equation*}
At horizon $T$, a separable \ac{FTRL} learner
with learning rate $\eta_T>0$, fixed across rounds, plays
$y_t=Q_h(\eta_TS_t)$. Definition~\ref{def:separable-regularizer}
ensures that the response is unique and continuous.

This class includes several familiar learning rules.
Negative entropy, $\theta(u)=u\log u$ with $0\log0=0$,
gives Hedge/MWU. The Euclidean kernel $\theta(u)=u^2/2$ gives
the projection of the scaled cumulative scores
onto the simplex. Negative Tsallis entropy uses
$\theta_q(u)=-u^q/[q(1-q)]$ for $q\in(0,1)$.
Appendix~\ref{app:learner-classes} derives their mean-based guarantees and the
corresponding learning-rate conditions.

% !TEX root = ../main.tex
\section{General-sum hardness: mixed control does not help}
\label{sec:hardness}

\subsection{Main result}
We first state the computational obstruction for
general-sum games. The optimizer may choose a
mixed action at every round, while the
benchmark is restricted to pure trajectories.
\begin{maintheorem}
\label{thm:main}
Fix $\beta\in(0,1]$ and a rational $0<\zeta<\beta/2$. There
is a polynomial-time computable family of zero-initialized
rational games with $n,m\le T^\beta$, $A\in[0,1]^{n\times m}$,
and $B\in[-1,1]^{n\times m}$, together with an
absolute constant $c>0$ with the following
property. For any fixed deterministic $\gamma_T$-mean-based
learner family satisfying $\gamma_TT=o(T^{1-\zeta})$, unless $\mathsf{P}=\mathsf{NP}$,
no polynomial-time optimizer returning an explicitly
encoded rational mixed trajectory can guarantee
\[
  R_{\mathrm{alg}}(T)\ge R_{\mathrm{pure}}^\star(T)-cT+o(T)
\]
on every game in the family at
all sufficiently large reduction horizons, where the
$o(T)$ error is uniform over game inputs.
\end{maintheorem}

Since every pure trajectory is a
mixed trajectory and $R_{\mathrm{pure}}^\star(T)\le R_{\mathrm{mix}}^\star(T)$, the
theorem also rules out the other
pure/mixed output--benchmark combinations (Corollary~\ref{cor:four-tasks} in Appendix~\ref{app:four-tasks}).

\begin{remark}[Fixed learner and mean-based promise]
\label{rem:learner-promise}
The learner family is fixed independently of the game instance, and the optimizer may be tailored to its known rule. The reduction uses only the quantitative mean-based promise in Theorem 1, without evaluating the learner or encoding a varying learning rule in the input.
\end{remark}
\begin{remark}[Computational model]
\label{rem:computational-model}
The optimizer is a uniform deterministic Turing
machine whose input consists of the explicitly
listed rational payoff matrices $A,B$ and the
horizon $T$. It may additionally receive a
rational bound $\bar\gamma_T\ge\gamma_T$ on the mean-based parameter,
promised to be
valid for the fixed learner. Any claimed
guarantee must hold for every valid supplied
bound. The hardness remains valid with this
information: at the reduction horizons, one
can supply $\bar\gamma_T:=1/\lceil T^\zeta\rceil$, which is
computable in polynomial time, has $O(\log
T)$ bits, and is valid for
all sufficiently large $T$ because $\gamma_T=o(T^{-\zeta})$.

Integers are encoded in binary, and rational
numbers by a binary numerator and a
positive binary denominator. Running time is measured
as a function of $T$ and the
total input encoding length, including any supplied
bound $\bar\gamma_T$; polynomial time means polynomial in
these quantities. An explicitly encoded rational mixed
trajectory lists every coordinate of every $x_t$,
and its total bit length is charged
to the optimizer's running time. To evaluate
the optimizer's output, we use $R_{\mathrm{alg}}(T)$ for
the reward $R(x,T)$ of the trajectory
$x$ it returns. Throughout, we use
$\mathsf{P}$ and $\mathsf{NP}$ for the standard
deterministic and nondeterministic polynomial-time complexity classes.
\end{remark}

\begin{remark}[Comparison with prior work]
\label{rem:prior-work}
Theorem~\ref{thm:main} holds for every fixed deterministic $\gamma_T$-mean-based
learner family with $\gamma_TT=o(T^{1-\zeta})$. This addresses a
question in the concluding section of \citet{assos2025intractability},
who prove linear additive hardness for Hedge/MWU
with pure optimizer outputs and ask what
happens for learners other than MWU. Our
optimizer may output mixed actions while the
benchmark is the best pure sequence, so
its task is easier and the hardness
result is stronger. Hedge/MWU with $\eta_T=T^{-(1-\alpha)}$
is covered: the negative-entropy case of
Appendix~\ref{app:learner-classes} with a rational $\zeta<\min(\alpha,\beta/2)$ recovers
their result from a different source
problem (Corollary~\ref{cor:setwise-hardness}); the same appendix treats
Euclidean, negative Tsallis, and log-barrier regularization.
\end{remark}

\subsection{Source problem: Independent Set on co-partite graphs}
\label{sec:souce-problem}
Following \citet[Definition~2.1]{dinursafra2005vertexcover}, a $(K,r)$-\emph{co-partite graph} is a
graph $G=([K]\times R,E)$ with $|R|=r$ such that,
for every $i\in[K]$, the $r$ vertices $\{i\}\times
R$ form a clique. Edges between different
such cliques are arbitrary. Thus $G$ has
$N:=Kr$ vertices, partitioned into $K$ cliques of
size $r$. We write $\operatorname{IS}(G)$ for the
maximum size of an independent set in
$G$. Since every independent set contains at
most one vertex from each clique, $\operatorname{IS}(G)\le
K$. We use the following constant-gap problem
as the source of our reduction.

\phantomsection
\begin{gapfact}
\label{fact:gap-is}
There exists an absolute constant integer
$r\ge 2$ such that it is
NP-hard, given a $(K,r)$-co-partite graph $G$
on $N=Kr$ vertices, to distinguish between
\begin{equation*}
  \mathrm{YES:}\qquad
  \operatorname{IS}(G)=\frac{N}{r},
  \qquad\qquad
  \mathrm{NO:}\qquad
  \operatorname{IS}(G)\le \frac{N}{128r}.
  \tag{GIS}
  \label{eq:source-gap}
\end{equation*}
\end{gapfact}
The hardness statement for~\eqref{eq:source-gap} specializes Theorem~2.1
of \citet{dinursafra2005vertexcover} to $h=2$
and $\varepsilon=1/128$. Appendix~\ref{app:dinur-safra} records the derivation.\footnote{This
constant-gap promise is where the probabilistically checkable
proofs (PCP) theorem enters our argument; see
\citet[Section~9]{dinursafra2005vertexcover}. A YES instance permits one mutually
compatible vertex choice per clique, whereas a
NO instance permits choices from only a
small fraction of the cliques. Since $r$
is fixed, the optimum independent-set densities are
separated by a positive constant. Our reduction transfers
this separation to rewards differing by $\Omega(T)$,
even after the constant losses and $o(T)$
errors in the analysis. Ordinary NP-hardness of
finding an exact maximum independent set would
not, by itself, provide the gap needed
for an additive $cT$ hardness result.}

\subsection{The game: activate a vertex, then harvest
it}
\label{sec:construction}
\label{sec:reduction}
Our reduction starts from the gap Independent
Set problem introduced in Section~\ref{sec:souce-problem}. In a
nutshell, our goal is to construct a
game in which approximating the pure benchmark
to within $cT$ would decide the gap
problem. During the interaction, the optimizer's choice
of which vertex to activate and harvest
mirrors the decision whether to include that
vertex in the independent set. Completeness is
comparatively direct: an independent set yields a
pure trajectory whose reward is proportional to
its size. Soundness is the main difficulty:
an arbitrary mixed trajectory may fractionally interleave
many actions at every round, and
we must still recover an independent
set whose size controls its reward.

Fix the constant $r$ from~\eqref{eq:source-gap} throughout, and
let $G=(V,\mathcal{E})$ be a simple undirected graph
on $N\ge2$ vertices; the game below is
defined for every such graph, and the
co-partite structure is used only for the
source gap. The optimizer has a setup
action $p$, together with an activation action
$a_v$ and a harvesting action $h_v$ for
each vertex $v$. The learner has a
neutral action $b_0$, together with a rewarding
action $b_v$ and a local alarm $d_v$
for each vertex, and an edge alarm
$d_{uv}$ for each $uv\in\mathcal{E}$. Hence $n=2N+1$ and
$m=2N+|\mathcal{E}|+1$. The optimizer's only nonzero payoff is
$A(h_v,b_v)=1$, so its round-$t$ reward is
\begingroup
\setlength{\abovedisplayskip}{4pt}
\setlength{\belowdisplayskip}{4pt}
\setlength{\abovedisplayshortskip}{0pt}
\setlength{\belowdisplayshortskip}{4pt}
\begin{equation*}
  r_t:=x_t^{\mathsf T}Ay_t=\sum_{v\in V}x_t(h_v)y_t(b_v).
  \tag{round reward}
  \label{eq:reward-identity}
\end{equation*}
\endgroup
In particular, setup and activation never directly
earn reward.

To express the score dynamics through the
optimizer's action counts, for $F\subseteq\mathcal{A}$ we track
the round mass $x_t(F):=\sum_{a\in F}x_t(a)$ and the
cumulative mass $X_t(F):=\sum_{s<t}x_s(F)$, writing $X_t(a)$ for a
single action. The completeness construction of Section~\ref{sec:completeness}
first plays the setup action $p$ for
$q_0$ rounds and then processes vertices in
activation and harvesting blocks whose lengths are
multiples of $k$; we call $q_0$ the
setup length and $k$ the block scale,
with $q_0=Nk$. We refer to the normalized
setup count $X_t(p)/N$ as the setup level;
after $q_0$ pure setup rounds it equals
$k$, and along an arbitrary mixed trajectory
it may continue to grow as setup
is interleaved with activation and harvesting.

The learner payoffs are specified by the
table below, where $uw\in\mathcal{E}$ indexes the last
column; they depend only on the
graph, and all entries lie in
$[-1,1]$ and have $O(\log N)$-bit encodings.
Since $S_0=0$, summing the payoffs gives
the score identities beside the table.
\begingroup
\setlength{\topsep}{4pt}
\setlength{\partopsep}{0pt}
\begin{center}
\begin{minipage}[c]{0.55\linewidth}
\small\setlength{\arraycolsep}{3pt}
\setlength{\abovedisplayskip}{2pt}
\setlength{\belowdisplayskip}{2pt}
\setlength{\abovedisplayshortskip}{0pt}
\setlength{\belowdisplayshortskip}{2pt}
\[
\begin{array}{c|cccc}
  & b_0 & b_w & d_w & d_{uw}\\ \hline
  p   & 0 & -1/N & -3/(2N) & -3/(2N)\\
  a_v & 0 & \mathbf 1\{v=w\} & \mathbf 1\{v=w\}
      & \mathbf 1\{v\in\{u,w\}\}\\
  h_v & 0 & -\mathbf 1\{v=w\} & 0 & 0
\end{array}
\]
\end{minipage}%
\hfill
\begin{minipage}[c]{0.44\linewidth}
\small
\setlength{\abovedisplayskip}{2pt}
\setlength{\belowdisplayskip}{2pt}
\setlength{\abovedisplayshortskip}{0pt}
\setlength{\belowdisplayshortskip}{2pt}
\begin{equation*}
\begin{aligned}
  S_t(b_0)&=0,\\
  S_t(b_v)&=X_t(a_v)-X_t(h_v)-\frac{X_t(p)}{N},\\
  S_t(d_v)&=X_t(a_v)-\frac{3X_t(p)}{2N},\\
  S_t(d_{uv})&=X_t(a_u)+X_t(a_v)-\frac{3X_t(p)}{2N}.
\end{aligned}
\tag{$\bigstar$}
\label{eq:score-identities}
\end{equation*}
\end{minipage}
\end{center}
\endgroup
Harvesting lowers the rewarding score without erasing
the activation count recorded by the
alarms. The local alarm limits how
much one vertex can pay, while
the edge alarm records that a
neighbor was activated. These permanent counts
will also constrain arbitrary fractional interleavings.

\subsection{Completeness: harvesting an independent set}
\label{sec:completeness}

An independent set allows us to process
its vertices one at a time. Each
activation block makes the current rewarding action
the score leader; harvesting then collects reward
and returns its score to zero, allowing
the next vertex to take the lead.
Independence keeps every edge alarm below
the current rewarding action during harvesting.

\begin{proposition}[Pure completeness]
\label{prop:complete}
Let \(k\) be a positive even integer and \(T=3Nk\).
For every independent set
$I\subseteq V$, there is a pure trajectory
$\pi$ satisfying
\begin{equation*}
  R(\pi,T)\ge |I|\frac{k}{2}-|I|(\gamma_TT+1)-m\gamma_TT.
\end{equation*}
\end{proposition}
For concision, we give the full proof
in Appendix~\ref{app:proof-complete}. The construction processes the vertices
of an independent set one at a
time. After the initial setup, activating a
vertex \(v\) raises the score of its
rewarding learner action \(b_v\) from \(-k\) to
\(k/2\). The optimizer then harvests \(v\), and
each harvesting round lowers this score by
one while earning reward whenever the learner
plays \(b_v\), so one activated vertex yields
about \(k/2\). Independence makes this mechanism repeatable:
when the optimizer moves to a new
vertex \(v\in I\), no previously activated vertex
is adjacent to \(v\), so no edge
alarm associated with \(v\) can compete with
\(b_v\), and \(b_v\) leads every other learner
action by more than the mean-based threshold
throughout most of the harvesting block.

\subsection{Soundness: extracting an independent set from mixed
play}
\label{sec:soundness}

Fix an arbitrary real mixed trajectory $x=(x_t)_{t<T}$.
A vertex can contribute substantial reward only
when its rewarding action is close to
the highest-scoring learner action. Comparing it with
the neutral action and the alarms then
forces its activation count to be large
and its neighbors' counts to be
small, which yields an independent set
as long as the setup level
$X_t(p)/N$ does not change too much.

We capture this condition by calling $v$
\emph{near-leading} at time $t$ if $S_t(b_v)\ge\max_{b\in\mathcal{B}}S_t(b)-\gamma_TT$. Comparing
$b_v$ with $b_0$, $d_v$, and each
incident edge alarm in \eqref{eq:score-identities} gives,
respectively, the following near-leading constraints (NLCs)
on the activation and harvesting counts:
\begingroup
\setlength{\abovedisplayskip}{2pt}
\setlength{\belowdisplayskip}{2pt}
\setlength{\abovedisplayshortskip}{2pt}
\setlength{\belowdisplayshortskip}{2pt}
\setlength{\jot}{0pt}
\begin{align*}
  X_t(a_v)&\ge\frac{X_t(p)}{N}+X_t(h_v)-\gamma_TT,\nonumber\\
  X_t(h_v)&\le\frac{X_t(p)}{2N}+\gamma_TT,\tag{NLCs}\label{eq:local-comparisons}\\
  X_t(a_u)&\le\frac{X_t(p)}{2N}+\gamma_TT-X_t(h_v)
  \qquad(u\in N_G(v)).\nonumber
\end{align*}
\endgroup
Here $N_G(v)$ is the set of neighbors
of $v$.

\paragraph{Extraction by setup intervals.}
We now turn the near-leading conditions into
an extraction rule that depends only on
the optimizer trajectory. Since $X_t(h_v)\ge0$, the~\ref{eq:local-comparisons} give
$X_t(a_v)\ge X_t(p)/N-\gamma_TT$ and $X_t(a_u)\le X_t(p)/(2N)+\gamma_TT$ for every
neighbor $u$: a near-leading vertex has activation
close to the full setup level,
while each neighbor has at most
about half. At every time $t$ with
$X_t(p)/N\ge1$, we therefore collect the qualifying vertices
in the set
\begingroup
\setlength{\abovedisplayskip}{2pt}
\setlength{\belowdisplayskip}{2pt}
\setlength{\abovedisplayshortskip}{2pt}
\setlength{\belowdisplayshortskip}{2pt}
\[
  J_t:=\left\{v\in V:\;
    X_t(a_v)\ge\tfrac{9X_t(p)}{10N}
    \;\text{and}\;
    X_t(a_u)\le\tfrac{3X_t(p)}{5N}
    \quad\forall u\in N_G(v)
  \right\}.
\]
\endgroup
The constants $9/10$ and $3/5$ relax the
ideal thresholds $1$ and $1/2$ to absorb
the additive $\gamma_TT$ error, so every near-leading
vertex belongs to $J_t$ whenever $X_t(p)/N\ge10(\gamma_TT+1)$.

At any fixed time, these thresholds already
exclude adjacent candidates. We cannot simply take
the union of $J_t$ over all times,
however. As the setup level $X_t(p)/N$ grows,
a previously large activation count may become
relatively small, allowing a neighbor to enter
a later candidate set. We therefore group
the rounds by setup intervals $[L,cL)$. We
choose $c$ small enough to preserve this
separation throughout an interval. If $v$ enters
$J_t$ in this interval, then $X_t(a_v)\ge9L/10$, and
this lower bound persists because cumulative counts
never decrease. A neighbor $u$ entering $J_{t'}$
later in the same interval would instead
require $X_{t'}(a_v)\le3X_{t'}(p)/(5N)<(3/5)cL$. The two bounds are incompatible
whenever $(3/5)c\le9/10$, so we take $c=3/2$. We
organize these intervals geometrically: for each integer
$j\ge0$, we use the level $L_j:=(3/2)^j$
and collect the candidates appearing while
the setup level lies in $[L_j,(3/2)L_j)$,
\[
  I_j:=\bigcup_{\{t:\,L_j\le \frac{X_t(p)}{N} <\frac{3}{2} L_j\}}J_t.
\]

Different setup intervals may reveal different independent
sets $I_j$, possibly of different sizes. Among
these extracted candidates, we return one of
maximum cardinality, denoted $I_{\mathrm{alg}}$, breaking ties in
favor of the smallest $j$; if no
interval contains a state, we return the
empty set. The rule uses only
cumulative optimizer counts, and Proposition~\ref{prop:sound} bounds
the reward in terms of $|I_{\mathrm{alg}}|$.
\begin{lemma}
\label{lem:independent}
Every mixed trajectory yields independent sets $I_j$,
regardless of learner behavior.
\end{lemma}
Figure~\ref{fig:independent-set-extraction} illustrates the extraction
rule and why each union $I_j$ remains independent.
\begingroup
\setlength{\intextsep}{6pt}
\setlength{\abovecaptionskip}{4pt}
\floatstyle{plain}
\restylefloat{figure}
\begin{figure}[ht]
  \centering
  \includegraphics[width=\linewidth]{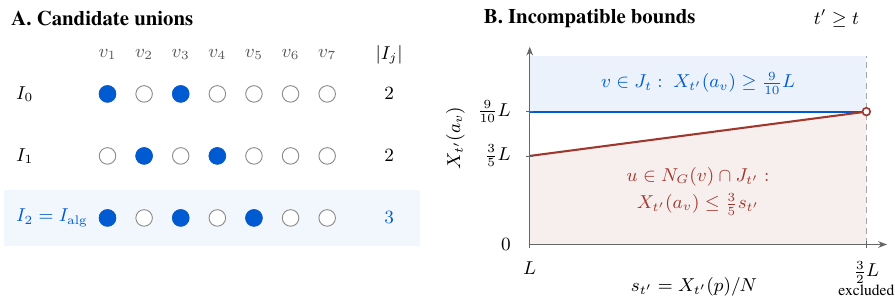}
  \caption{\small Extraction by setup intervals.
    (A) Each row is the union $I_j$ of the candidate sets over one
    setup interval $[L_j,\frac32L_j)$; by Lemma~\ref{lem:independent},
    its blue vertices are pairwise non-adjacent. Extraction returns
    a largest row, $I_{\mathrm{alg}}$.
    (B) Why one interval $[L,\frac32L)$ cannot contain two adjacent
    candidates, where $s_{t'}:=X_{t'}(p)/N$. Once $v$ enters $J_t$ in
    this interval, its activation count stays at least $\frac9{10}L$
    at every later time $t'\ge t$ (blue). A neighbor $u$ can enter
    $J_{t'}$ only if $X_{t'}(a_v)\le\frac35s_{t'}$ (red). The two
    regions meet only at the excluded endpoint $s_{t'}=\frac32L$, so
    $3/2$ is the largest ratio that keeps them apart.}
  \label{fig:independent-set-extraction}
\end{figure}
\endgroup

\begin{proposition}[Mixed soundness and extraction]
\label{prop:sound}
For every real mixed trajectory $x$, the
extraction rule above returns an independent set
$I_{\mathrm{alg}}$ satisfying
\begin{equation*}
  R(x,T)\le6|I_{\mathrm{alg}}|\frac{T}{N}
    +6N(\gamma_TT+1)+\gamma_TT.
\end{equation*}
For explicitly encoded rational trajectories, extraction runs
in polynomial time and uses neither
the learner nor its mean-based parameter.
\end{proposition}

For concision, we give the full proofs
of Lemma~\ref{lem:independent} and Proposition~\ref{prop:sound} in Appendix~\ref{app:proof-sound}. By
\eqref{eq:reward-identity}, round $t$ pays the harvesting mass
$x_t(h_v)$ weighted by $y_t(b_v)$. Vertices that are
not near-leading have $y_t(b_v)\le\gamma_T$ by mean-basedness and
contribute at most $\gamma_T$ per round. At
small setup levels, $X_t(p)/N<10(\gamma_TT+1)$, the~\ref{eq:local-comparisons} bound the
total harvesting mass of each near-leading vertex
by $6(\gamma_TT+1)$, so these rounds contribute at
most $6N(\gamma_TT+1)$. At large setup levels, every
near-leading vertex lies in $J_t$, hence in
the independent set $I_j$ of its geometric
interval, where the same constraints keep its
harvesting mass below $2L_j$; summing $2L_j|I_j|\le2L_j|I_{\mathrm{alg}}|$
over the intervals up to the
maximal setup level $T/N$ gives $6|I_{\mathrm{alg}}|T/N$.

\subsection{Horizon choice and hardness transfer}
\label{sec:hardness-transfer}

The horizon must make both the
score-gap losses and the coordinatewise probability
leakage sublinear. For the fixed rational
$\zeta$ in Theorem~\ref{thm:main}, choose $k:=2\lceil N^{2/\zeta-1}\rceil$,
$q_0:=Nk$, and $T:=q_0+2Nk=3Nk$. These choices give
\begin{equation*}
  T=\Theta(N^{2/\zeta}),
  \qquad N=\Theta(T^{\zeta/2}),
  \qquad k=\Theta(T^{1-\zeta/2}).
  \tag{Reduction scales}
  \label{eq:reduction-scales}
\end{equation*}
Appendix~\ref{app:horizon-bookkeeping} checks that the horizon and the
encodings are polynomial, that $n,m\le T^\beta$ for
large instances, and that the error terms
of Propositions~\ref{prop:complete} and~\ref{prop:sound} are $o(T)$, so that
$R_{\mathrm{pure}}^\star(T)\ge\operatorname{IS}(G)T/(6N)-o(T)$ and $R(x,T)\le6|I_{\mathrm{alg}}|T/N+o(T)$. It then completes the
proof of Theorem~\ref{thm:main} with $c=1/(24r)$: an optimizer
achieving the claimed guarantee would give
$|I_{\mathrm{alg}}|\ge\operatorname{IS}(G)/36-cN/6-o(N)$, which exceeds $N/(64r)$ on YES
instances of \eqref{eq:source-gap} while NO instances
have no independent set larger than
$N/(128r)$, so the optimizer would decide
the promise problem in polynomial time.

% !TEX root = ../main.tex
\section{Zero-sum games: pure play keeps the full
surplus}
\label{app:zero-sum}

Section~\ref{sec:hardness} showed that mixed play does not
remove the computational obstruction in general-sum games.
In zero-sum games, the game value gives
us a useful starting point: if the
learner has sublinear external regret and can
be simulated efficiently, a pure myopic best
response already approximates the optimal mixed trajectory
within a sublinear additive error. We begin
with this observation, then explore surplus against
separable \ac{FTRL} in Sections~\ref{app:zero-sum-passive} and~\ref{app:zero-sum-active}. A fixed
max-min strategy earns passive surplus while the
learner still places mass on suboptimal
responses; the learner's discrete updates can
provide a second source of surplus.

Throughout this section, we take $B=-A$ and
optimizer payoffs in $[-1,1]$. The scores therefore
evolve as $S_{t+1}=S_t-A^{\mathsf T}x_t$, starting from $S_0=0$.
We index actions by rows or columns
and use a fixed order to break
ties. A pure trajectory $\pi=(\pi_0,\ldots,\pi_{T-1})$ has $\pi_t=e_{\ell_t}$
for every $t$. Alongside the trajectory
benchmarks, we use the game value
\begingroup
\setlength{\abovedisplayskip}{1pt}
\setlength{\belowdisplayskip}{1pt}
\setlength{\abovedisplayshortskip}{1pt}
\setlength{\belowdisplayshortskip}{1pt}
\begin{equation*}
  \operatorname{Val}(A):=\max_{u\in\Delta(\mathcal{A})}\min_{v\in\Delta(\mathcal{B})}u^{\mathsf T}Av
  =\min_{v}\max_{u}u^{\mathsf T}Av.
\end{equation*}
\endgroup
We state all computational guarantees for rational
game data and explicit trajectories, with running
time polynomial in the horizon $T$, the
input encoding length, and the required numerical
precision. As in the hardness reduction, producing
the trajectory is part of the computation.

\paragraph{Warm-up: a pure myopic best response.}
\label{app:zero-sum-approximation}
Our first step is to compare a
simple myopic rule with the offline benchmarks.
When $B=-A$, the external-regret guarantee from Section~\ref{sec:model-result}
gives $R(x,T)\le\min_{y\in\Delta(\mathcal{B})}\sum_{t<T}x_t^{\mathsf T}Ay +\operatorname{Regret}_{\mathsf L}(T)$ for every
mixed trajectory $x$. At each round $t$,
the optimizer chooses a pure best response
$\pi_t:=e_{\ell_t}$, $\ell_t\in\operatorname*{arg\,max}_{\ell\in\mathcal{A}}e_\ell^{\mathsf T}Ay_t$, breaking ties by
the fixed action order. Proposition~\ref{prop:app-pure-best-response} in
Appendix~\ref{app:proof-best-response} shows that the resulting pure
trajectory $\pi_{\mathrm{BR}}$ satisfies $R(\pi_{\mathrm{BR}},T)\ge T\operatorname{Val}(A)$ and
\begingroup
\setlength{\abovedisplayskip}{1pt}
\setlength{\belowdisplayskip}{1pt}
\setlength{\abovedisplayshortskip}{1pt}
\setlength{\belowdisplayshortskip}{1pt}
\[
    R(\pi_{\mathrm{BR}},T)
    \ge R_{\mathrm{mix}}^\star(T)-\operatorname{Regret}_{\mathsf L}(T)
    \ge R_{\mathrm{pure}}^\star(T)-\operatorname{Regret}_{\mathsf L}(T).
\]
\endgroup
The comparison follows from two uses of
the game value. By the minimax theorem,
every pure best response earns at least
$\operatorname{Val}(A)$. The regret bound, in turn, caps
every trajectory, pure or mixed, at
$T\operatorname{Val}(A)+\operatorname{Regret}_{\mathsf L}(T)$ because the average optimizer
action is itself a mixed strategy.

\subsection{The potential decomposition}
\label{sec:potential-decomposition}

To analyze surplus above \(T\operatorname{Val}(A)\), we specialize to separable FTRL from Section~\ref{sec:ftrl-preliminaries} and use the conjugate regularizer and its scaled potential,
\begingroup
\setlength{\abovedisplayskip}{-2pt}
\setlength{\belowdisplayskip}{1pt}
\setlength{\abovedisplayshortskip}{-2pt}
\setlength{\belowdisplayshortskip}{1pt}
\[
  h^*(z):=\max_{y\in\Delta(\mathcal{B})}\{\langle z,y\rangle-h(y)\},
  \qquad
  \Phi_{\eta_T}(S):=\frac{1}{\eta_T}h^*(\eta_TS),
\]
Here $\nabla h^*=Q_h$. We measure the contribution
of each discrete update through $D_{h^*}(u,v):=h^*(u)-h^*(v)-\langle
Q_h(v),u-v\rangle$, the Bregman divergence of $h^*$.
With these quantities in place, the
decomposition of \citet[Theorem~3.3]{su2026nocoin} writes the reward
as a terminal potential difference plus
a sum of nonnegative Bregman terms:
\begin{equation*}
  R(x,T)
  =\Phi_{\eta_T}(0)-\Phi_{\eta_T}(S_T)
   +\frac{1}{\eta_T}\sum_{t=0}^{T-1}
       D_{h^*}(\eta_TS_{t+1},\eta_TS_t).
  \tag{\acs{FTRL}-Reward}
  \label{eq:app-disc-cont-used}
\end{equation*}
\endgroup
To interpret the terminal term, let us
follow a fixed max-min strategy $x^\star\in\Delta(\mathcal{A})$ in
continuous time. We write $\operatorname{supp}(x^\star):=\{\ell:x^\star_\ell>0\}$ for
its support. The corresponding reward is
\begin{equation}
  R_{\mathrm{cont}}(x^\star,T)
  =\int_0^T (x^\star)^{\mathsf T}A
         Q_h(-\eta_T sA^{\mathsf T}x^\star)\,ds
  =\Phi_{\eta_T}(0)-\Phi_{\eta_T}(-TA^{\mathsf T}x^\star),
  \label{eq:app-zero-sum-continuous-reward}
\end{equation}
The equality follows from the path-independence identity
of \citet[Lemma~3.1]{su2026nocoin}. Because a max-min strategy earns at
least the value against every learner distribution,
$R_{\mathrm{cont}}(x^\star,T)\ge T\operatorname{Val}(A)$.
We define the optimizer’s \emph{surplus} at horizon \(T\) as \(R(x,T)-T\operatorname{Val}(A)\), the excess of its cumulative reward over the game-value baseline. For fixed max-min play,
this passive surplus accumulates while
the learner places mass on actions outside
the best-response set $\operatorname{br}(x^\star):=\operatorname*{arg\,min}_{j\in\mathcal{B}}(x^\star)^{\mathsf T}Ae_j$.
We first show that pure actions can
preserve this passive surplus, then use the
Bregman terms to obtain active surplus.

\subsection{Passive surplus survives terminal rounding}
\label{app:zero-sum-passive}
We first recall the continuous-time surplus of
\citet[Theorems~4.1 and~4.2]{su2026nocoin}. Let
$N_{\mathrm{sub}}:=m-|\operatorname{br}(x^\star)|$ count the learner
actions that are not best responses to the max-min strategy $x^\star$.

\begin{proposition}[\citet{su2026nocoin}]
\label{prop:app-passive-ftrl-baseline}
Fix the game $A$, a max-min strategy $x^\star$, and the separable regularizer
$h$. If $N_{\mathrm{sub}}>0$, there is a constant $c_{A,h}>0$ such
that, for sufficiently large $\eta_TT$, it holds that
\begin{equation*}
  R_{\mathrm{cont}}(x^\star,T)
  \ge T\operatorname{Val}(A)
      +c_{A,h}{N_{\mathrm{sub}}}/{\eta_T}.
\end{equation*}
\end{proposition}
\begin{remark}
    The constant $c_{A,h}$ and the threshold on $\eta_TT$ depend only on
    the game $A$, strategy $x^\star$, and regularizer $h$.
    \citet[Theorem~4.1]{su2026nocoin} also establish an upper bound on
    the continuous-time surplus from fixed max-min play proportional to
    $N_{\mathrm{sub}}/\eta_T$. When $N_{\mathrm{sub}}=0$, every learner
    action earns $\operatorname{Val}(A)$ against $x^\star$, so
    $R_{\mathrm{cont}}(x^\star,T)=T\operatorname{Val}(A)$ for all $T\ge0$.
\end{remark}

This surplus accrues while the learner removes probability from
suboptimal actions. To preserve it with pure play, we use
\emph{terminal balance}: round the target counts $Tx_\ell^\star$ to
integers $q_\ell\in\{\lfloor Tx_\ell^\star\rfloor,
\lceil Tx_\ell^\star\rceil\}$ summing to $T$, then play each action
$\ell$ exactly $q_\ell$ times. Lemma~\ref{lem:app-terminal-rounding} in Appendix \ref{app:proof-passive}
constructs these counts in polynomial time.

\begin{maintheorem}[Pure passive extraction]
\label{thm:app-pure-passive}
For every max-min strategy $x^\star$ and horizon $T$, there is a
terminally balanced deterministic pure trajectory $\pi$ satisfying
\begin{equation}
\label{eq:app-pure-passive-main}
  R(\pi,T)\ge R_{\mathrm{cont}}(x^\star,T)
             -|\operatorname{supp}(x^\star)|.
\end{equation}
\end{maintheorem}
\begin{remark}
    Once the terminally balanced counts \(q_\ell\) are fixed, the lower bound in the theorem above holds for every ordering that plays each action \(\ell\) exactly \(q_\ell\) times by horizon \(T\). 
\end{remark}
Appendix~\ref{app:proof-passive} proves the theorem and gives Algorithm \ref{alg:app-passive-terminal-greedy}, a polynomial-time construction. Combining the theorem with Proposition~\ref{prop:app-passive-ftrl-baseline} yields the following corollary.

\begin{corollary}
\label{cor:pure-passive-inverse-rate}
Fix the game $A$, a max-min strategy $x^\star$, and the regularizer $h$.
If $N_{\mathrm{sub}}>0$, then, for sufficiently large $\eta_TT$, there is a
polynomial-time computable terminally balanced pure trajectory $\pi$ satisfying
\begin{equation}
\label{eq:app-pure-passive-inverse-rate}
  R(\pi,T)\ge T\operatorname{Val}(A)
             +c_{A,h}\frac{N_{\mathrm{sub}}}{\eta_T}
             -|\operatorname{supp}(x^\star)|.
\end{equation}
\end{corollary}

\subsection{Active surplus survives prefix balancing}
\label{app:zero-sum-active}
When $N_{\mathrm{sub}}=0$, as in Rock--Paper--Scissors, playing
$x^\star$ earns no passive surplus. \citet[Theorem~4.3]{su2026nocoin}
recover active surplus by alternating two mixed strategies
averaging to $x^\star$. For strongly convex separable \ac{FTRL}
and sufficiently small $\eta_T$, they obtain
$\Omega(\eta_TT/\operatorname{poly}(n,m))$ surplus with high
probability for independent uniform $[-1,1]$ payoff entries.
We retain this gain from the learner's delayed response using pure actions.

We require the same assumption from \citet[Assumption~1]{assos2024maximizing} and
\citet[Assumption~1]{su2026nocoin} for the active surplus.
\begin{assumption}[Non-identification on the max-min support]
\label{ass:app-non-identification}
Some max-min strategy $x^\star\in\Delta(\mathcal A)$ has
distinct $i,j\in\operatorname{br}(x^\star)$ such that
$A_{\ell i}\ne A_{\ell j}$ for some
$\ell\in\operatorname{supp}(x^\star)$.
\end{assumption}
Terminal balance controls only the final counts. For the active surplus, we need to strengthen
it to prefix balance to control the learner throughout pure play.
We call a pure trajectory \emph{prefix balanced} if
\begin{equation*}
  \lfloor tx_\ell^\star\rfloor\le X_t(\ell)\le\lceil tx_\ell^\star\rceil
  \qquad\text{for every }\ell\in\mathcal A\text{ and }t\le T.
  \tag{Prefix Quotas}
  \label{eq:app-prefix-quotas}
\end{equation*}

\Needspace{14\baselineskip}
\begin{maintheorem}[Balanced pure active extraction]
\label{thm:app-pure-active}
Let $A\in[-1,1]^{n\times m}$ and fix $x^\star,i,j$
satisfying Assumption~\ref{ass:app-non-identification}.
Let the learner use \ac{FTRL} with a separable regularizer $h$
that is $\mu$-strongly convex with respect to $\|\cdot\|$,
and let $\|\cdot\|_*$ denote the dual norm.
Fix $\delta\in(0,1/m)$ and set
\[
    \sigma_{ij}^2(x^\star)
  =\sum_{\ell\in\operatorname{supp}(x^\star)}
       x_\ell^\star(A_{\ell i}-A_{\ell j})^2, \qquad
  M:=\sup_{u\in[\delta,1-\delta]}\theta''(u),
  \qquad
  L_A:=\max_{\ell\in\operatorname{supp}(x^\star)}
      \|A^{\mathsf T}e_\ell\|_*.
\]
Choose step size $\eta_T\le
  \frac{\mu}{c_{\|\cdot\|}(|\operatorname{supp}(x^\star)|+1)L_A}
  \left(\frac1m-\delta\right)$
where $c_{\|\cdot\|}=\sup_{u\ne0}\|u\|_\infty/\|u\|$.
Then there exists a deterministic prefix balanced pure trajectory
$\pi$ satisfying
\begin{equation}
\label{eq:app-pure-active-main}
  R(\pi,T)\ge T\operatorname{Val}(A)
  +\frac{\eta_T T}{4M}\,\sigma_{ij}^2(x^\star)
  -\frac{\eta_T|\operatorname{supp}(x^\star)|^2}{2\mu}L_A^2.
\end{equation}
\end{maintheorem}

\paragraph{Proof idea.}
For an update $v_t=A^{\mathsf T}x_t$, suppose two distinct coordinates $i,j$ of
$Q_h(\eta_T(S_t-rv_t))$ stay in $[\delta,1-\delta]$ for $0\le r\le1$.
Integrating the payoff-drop estimate of
\citet[Lemma~F.8]{su2026nocoin} gives
$\frac1{\eta_T}D_{h^*}(\eta_TS_{t+1},\eta_TS_t)
  \ge \frac{\eta_T}{4M}(v_{t,i}-v_{t,j})^2.$
Distinct payoffs thus yield a positive bonus. 
The details are proved in Proposition~\ref{prop:app-active-ftrl-baseline} in the appendix.

The mixture $x^\star$ ties the selected responses $i,j$, whereas pure
rows expose their payoff differences. Their average squared difference
$\sigma_{ij}^2(x^\star)>0$ measures the available bonus. The active
surplus requires both selected best responses $i$ and $j$ to retain
probability at least $\delta$ throughout each update.

Lemma \ref{lem:app-cumulative-rounding}  applies cumulative rounding \citep{golz2025apportionment} to sample a sequence of action indices \((\ell_0,\ldots,\ell_{T-1})\). At round \(t\), the optimizer plays the pure action \(\pi_t=e_{\ell_t}\). Every sampled sequence satisfies the prefix quotas: in its first \(t\) rounds, each action \(\ell\) is played between \(\lfloor tx_\ell^\star\rfloor\) and \(\lceil tx_\ell^\star\rceil\) times. Moreover, the sampling preserves the probability of each action at every round $\mathbb P(\ell_t=\ell)=x_\ell^\star, \text{ for every }\ell\in\mathcal A,\ t<T$.

The quotas keep $S_t$ close to $-tA^{\mathsf T}x^\star$, the score
of fixed max-min play. Along this reference trajectory, each best response
has probability at least $1/m$. Under the theorem's step-size choice,
Lemma~\ref{lem:app-shadowing} shows that both selected coordinates
$i,j$ remain in $[\delta,1-\delta]$ throughout each pure update.
Lemma~\ref{lem:app-one-step-bonus} therefore supplies a bonus in every
round. Taking expectation
over the cumulative-rounding randomness, using
$\mathbb{P}(\ell_t = \ell) = x_\ell^\star$,
and summing over all rounds gives
\begin{align*}
\mathbb{E}\left[
\frac{1}{\eta_T}
\sum_{t=0}^{T-1}
D_{h^*}(\eta_T S_{t+1}, \eta_T S_t)
\right] \geq
\frac{\eta_T}{4M}
\sum_{t=0}^{T-1}
\sum_{\ell \in \operatorname{supp}(x^\star)}
x_\ell^\star (A_{\ell i} - A_{\ell j})^2 =
\frac{\eta_T T}{4M}\sigma_{ij}^2(x^\star).
\end{align*}

To combine this bonus with the terminal contribution in
\eqref{eq:app-disc-cont-used}, we compare the pure terminal score $S_T$
with $-TA^{\mathsf T}x^\star$. Define their difference as $\Delta S
=S_T+TA^{\mathsf T}x^\star
=-\sum_\ell
\bigl(X_T(\ell)-Tx_\ell^\star\bigr)A^{\mathsf T}e_\ell.$ We need to control this difference's influence to the potential. Smoothness gives that
\begin{align*}
    \Phi_{\eta_T}(S_T)
\le \Phi_{\eta_T}(-T A^\top x^*) + \underbrace{
\left\langle Q_h(-\eta_T T A^\top x^*), \Delta S \right\rangle
}_{\text{Linear correction}}+ \underbrace{
\frac{\eta_T}{2\mu}\|\Delta S\|_*^2
}_{\text{Quadratic remainder bound}}.
\end{align*}
Since \(\mathbb P(\ell_t=\ell)=x_\ell^\star\) at every round, we have \(\mathbb E[X_T(\ell)]=Tx_\ell^\star\) and hence \(\mathbb E[\Delta S]=0\). The linear correction therefore vanishes in expectation.
The prefix quotas bound each terminal count discrepancy \(\lvert X_T(\ell)-Tx_\ell^\star\rvert\) by \(1\), so \(\|\Delta S\|_*\le|\operatorname{supp}(x^\star)|L_A\). Thus, the quadratic remainder is bounded by
$\frac{\eta_T}{2\mu}\|\Delta S\|_*^2
\le
\frac{\eta_T|\operatorname{supp}(x^\star)|^2}{2\mu}L_A^2.$

Taking expectation in the smoothness comparison, the linear correction vanishes. Hence the terminal contribution satisfies
\[
\begin{aligned}
\mathbb E\!\left[\Phi_{\eta_T}(0)-\Phi_{\eta_T}(S_T)\right]
&\ge
\Phi_{\eta_T}(0)-\Phi_{\eta_T}(-TA^{\mathsf T}x^\star)
-\frac{\eta_T|\operatorname{supp}(x^\star)|^2}{2\mu}L_A^2\\
&=R_{\mathrm{cont}}(x^\star,T)
-\frac{\eta_T|\operatorname{supp}(x^\star)|^2}{2\mu}L_A^2.
\end{aligned}
\]

Combining the expected bonus and terminal contribution gives
\eqref{eq:app-pure-active-main} in expectation. Since every sampled
sequence satisfies the prefix quotas, at least one such sequence also
attains this reward bound. Lemma~\ref{lem:app-cost-derandomization}
constructs such a sequence deterministically: it chooses rounding
outcomes so that the linear count corrections from the terminal
contribution and the bonus sum to a nonnegative value. These choices
preserve the prefix quotas and therefore give a deterministic pure
trajectory satisfying \eqref{eq:app-pure-active-main}.
Appendix~\ref{app:proof-active} gives the proof,
Appendix~\ref{app:algorithmic} the polynomial-time construction, and
Appendix~\ref{app:random-games} the quantitative random-game specialization.

% !TEX root = ../main.tex
\Needspace{12\baselineskip}

\section{Discussion}
\label{sec:discussion}

The contrast in our title can be
understood through cumulative action counts. These counts
determine the learner's cumulative payoff scores, although
a general mean-based learner may also depend
on the ordered history. In general-sum games,
the reduction of Section~\ref{sec:hardness} shows that, on
a constructed family of games, profitable play
encodes a combinatorial choice. The extraction rule
reads that choice off the counts whether
they are integral or fractional; mixing enlarges
the search space without weakening the argument.
In zero-sum games, the value and regret
bound place both trajectory benchmarks within
$\operatorname{Regret}_{\mathsf L}(T)$ of $T\operatorname{Val}(A)$. Our surplus constructions in
Section~\ref{app:zero-sum} choose pure schedules whose counts approximate
a target mixture, either terminally or at
every prefix, within one unit per action.

Both results have a definite scope. The
hardness is worst case: the instances are
constructed, their size grows with the horizon
as $n,m\le T^\beta$, and the learner must be
mean-based at the rate $\gamma_TT=o(T^{1-\zeta})$. Appendix~\ref{app:learner-classes} gives
sufficient learning-rate conditions for separable \ac{FTRL}.
The positive results use progressively more structure:
the myopic best-response guarantee needs only a
regret bound, and its computation requires efficient
learner simulation. The passive and active surplus
guarantees use the regularizer explicitly and, for
active surplus, the regularity, step-size, and non-identification
conditions of Theorem~\ref{thm:app-pure-active}.

Appendix~\ref{app:open-questions} collects four directions: local optimality of
controlled trajectories, proximal-regret learners, hardness beyond a
prescribed horizon, and smoothed instances. We highlight
two. First, a proximal-regret guarantee \citep{cai2025proximal} does
not by itself remove the barrier: under
their comparator assumptions, MWU satisfies the entropy-based
Bregman form of this guarantee while remaining
in our hard class (Appendix~\ref{app:extended-related-work}). The open
case is recursive Euclidean projected gradient descent,
whose attainable reward and planning complexity are
separate questions (direction~(B)). Second, Theorem~\ref{thm:main} concerns exactly
specified payoffs: the score identities~\eqref{eq:score-identities} hold with
equality. Motivated by smoothed analysis \citep{spielman2004smoothed}, we
consider independent uniform payoff perturbations in $[-\tau,\tau]$.
Along any fixed optimizer trajectory, these shift
each cumulative score by at most $\tau T$.
How much perturbation can the $\Omega(T)$ gap
withstand, and can larger perturbations permit efficient
approximation? Direction~(D) makes this question precise.

% !TEX root = ../main.tex
\section*{AI Disclosure}
\label{sec:ai-disclosure}

The authors obtained the mathematical results in
July 2026. An initial proof, spanning approximately
50 pages, used a reduction from $(1,2)$-Max
TSP but required an additional assumption on
the mean-based learner. To understand whether this
restriction came from the learning model or
the choice of source problem, we explored
alternative sources of strong $\mathsf{NP}$-hardness. This search led
us to the gap Independent Set problem
in~\eqref{eq:source-gap}, used in the present reduction.

Generative AI tools (Claude, Anthropic;
Codex, OpenAI) assisted with manuscript preparation, including
language editing, organization, \LaTeX{} typesetting, and checks
of notation, references, and consistency. The authors
take full responsibility for the correctness of
all statements and proofs and for the
originality and attribution of the content.
\clearpage

\printbibliography[heading=bibintoc]
\clearpage
\appendix
% !TEX root = ../main.tex
\begingroup
\raggedbottom
\setlength{\parskip}{3pt}
\section*{Proof roadmap}
\addcontentsline{toc}{section}{Proof roadmap}
\fontsize{10}{11.5}\selectfont

Let's read Figure~\ref{fig:hardness-proof-roadmap} from top to bottom.
Appendix~\ref{app:dinur-safra} describes the Independent Set gap from
\citet{dinursafra2005vertexcover}: YES graphs have a large independent
set, while NO graphs have only small
ones. The game of Section~\ref{sec:construction} encodes vertex
choices through activation and harvesting, and its
score identities support two complementary arguments.

Completeness (Appendix~\ref{app:proof-complete}) turns an independent set into
a valuable pure schedule. Soundness (Appendix~\ref{app:proof-sound}) shows
that any valuable mixed trajectory must reveal
a large independent set, despite fractional
interleavings. Geometric setup intervals ensure independence;
the mean-based condition bounds the reward in
terms of the extracted size. The
extraction uses only the optimizer's trajectory.

Appendix~\ref{app:horizon-bookkeeping} chooses a polynomial horizon with sublinear
losses and completes the reduction in Theorem~\ref{thm:main}:
the claimed approximation would distinguish YES from
NO graphs. Appendix~\ref{app:four-tasks} gives the other output--benchmark
combinations, and Appendix~\ref{app:learner-classes} establishes mean-based guarantees for separable \ac{FTRL}.
The independent zero-sum proofs in Appendix~\ref{app:zero-sum-proofs}
use regret and minimax for myopic
play, terminal rounding for passive surplus,
and prefix balancing for active surplus.

\par\smallskip
\noindent\begin{minipage}{\linewidth}
\centering
\begin{tikzpicture}[
  every node/.style={font=\fontsize{9}{10.5}\selectfont,align=left},
  proofbox/.style={draw=black!45,line width=.6pt,rounded corners=2pt,
    inner xsep=2.5mm,inner ysep=.6mm,anchor=north},
  fullbox/.style={proofbox,text width=13.9cm,fill=black!3},
  purebox/.style={proofbox,text width=6.45cm,fill=teal!5},
  mixedbox/.style={proofbox,text width=6.45cm,fill=blue!4},
  flow/.style={-{Stealth},line width=.75pt,draw=black!70}
]
  \node[fullbox] (source) at (0,0) {
    \textbf{Source gap: \eqref{eq:source-gap}}\\
    Large independent sets in YES graphs; small
    ones in NO graphs.
  };
  \node[fullbox,below=1.5mm of source] (game) {
    \textbf{Game and score identities: Section~\ref{sec:construction}}\\
    Setup, activation, and harvesting encode the graph
    in cumulative payoffs.
  };
  \node[purebox,below=1.5mm of game.south west,anchor=north west] (complete) {
    \textbf{Pure completeness: Proposition~\ref{prop:complete}}\\
    An independent set $I$ gives a pure
    schedule with reward $|I|T/(6N)-o(T)$ at the reduction
    horizons.
  };
  \node[mixedbox,below=1.5mm of game.south east,anchor=north east] (independent) {
    \textbf{Independent extraction: Lemma~\ref{lem:independent}}\\
    Group candidate vertices by geometric setup intervals;
    each union $I_j$ is independent for every
    mixed trajectory.
  };
  \node[mixedbox,below=1.5mm of independent] (sound) {
    \textbf{Mixed soundness: Proposition~\ref{prop:sound}}\\
    Extract a largest $I_j$ in polynomial time.
    The mean-based condition bounds reward by $6|I_{\mathrm{alg}}|T/N+o(T)$.
  };
  \node[fullbox] (horizon) at ([yshift=-1.5mm]sound.south -| game.center) {
    \textbf{Horizon and errors: Appendix~\ref{app:horizon-bookkeeping}}\\
    Polynomial horizon and encodings; the $\gamma_T$ error terms are
    $o(T)$.
  };
  \node[fullbox,below=1.5mm of horizon] (result) {
    \textbf{Hardness: Theorem~\ref{thm:main}}\\
    An accurate mixed optimizer against a
    pure benchmark would decide the gap.
  };
  \draw[flow] (source.south) -- (game.north);
  \draw[flow] (game.south -| complete.north) -- (complete.north);
  \draw[flow] (game.south -| independent.north) -- (independent.north);
  \draw[flow] (independent.south) -- (sound.north);
  \draw[flow] (complete.south) -- (complete.south |- horizon.north);
  \draw[flow] (sound.south) -- (sound.south |- horizon.north);
  \draw[flow] (horizon.south) -- (result.north);
\end{tikzpicture}

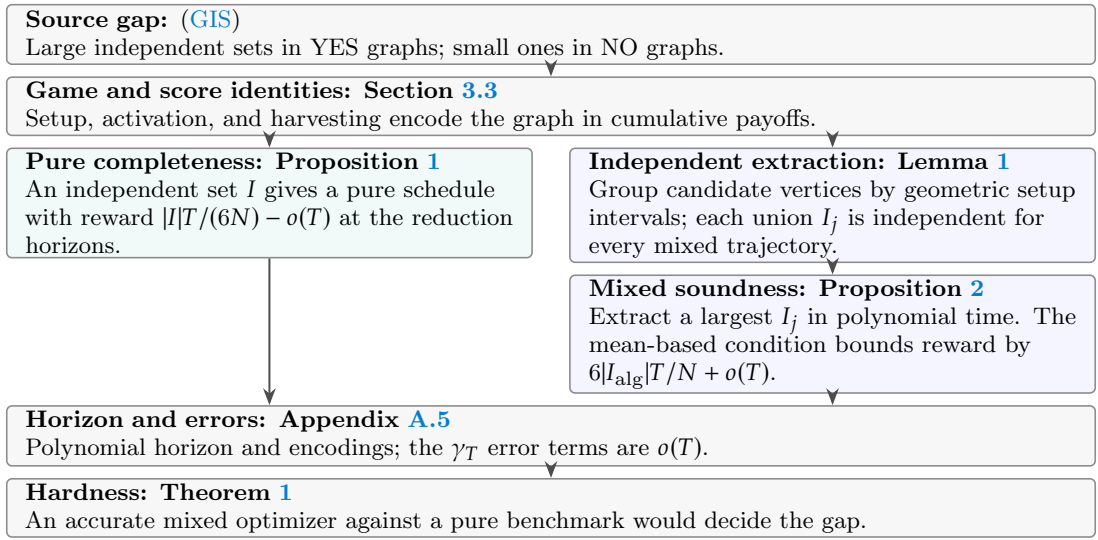
\captionof{figure}{\small Dependencies in the general-sum hardness proof. The
two branches establish reward bounds; the horizon
choice turns them into a gap reduction.}
\label{fig:hardness-proof-roadmap}
\end{minipage}
\par
\clearpage
\endgroup

% !TEX root = ../main.tex
\section{Omitted proofs for Section~\ref{sec:hardness}}
\label{app:hardness-proofs}

\subsection{Other output--benchmark combinations}
\label{app:four-tasks}

The optimizer may output mixed actions but
is compared only with the pure benchmark.
The other pure/mixed output--benchmark combinations follow immediately.
\begin{corollary}
\label{cor:four-tasks}
Under the assumptions of Theorem~\ref{thm:main}, no
polynomial-time optimizer can satisfy any one
of the following guarantees on every
game in the family at all
sufficiently large reduction horizons, unless $\mathsf{P}=\mathsf{NP}$:
\begin{align*}
  R(x,T)&\ge R_{\mathrm{pure}}^\star(T)-cT+o(T),
  & R(x,T)&\ge R_{\mathrm{mix}}^\star(T)-cT+o(T),\\
  R(\pi,T)&\ge R_{\mathrm{pure}}^\star(T)-cT+o(T),
  & R(\pi,T)&\ge R_{\mathrm{mix}}^\star(T)-cT+o(T).
\end{align*}
Here $x$ denotes the optimizer's mixed output
trajectory in the first row, and $\pi$
its pure output trajectory in the second.
\end{corollary}
\begin{proof}
The first guarantee is ruled out by
Theorem~\ref{thm:main}. The remaining cases follow because every
pure trajectory is a mixed trajectory and
$R_{\mathrm{pure}}^\star(T)\le R_{\mathrm{mix}}^\star(T)$.
\end{proof}

\subsection{The source gap: specialization of Dinur and
Safra}
\label{app:dinur-safra}

The source problem~\eqref{eq:source-gap} is a direct specialization
of the gap problem introduced by \citet{dinursafra2005vertexcover}.
To make the source of the constant
gap explicit, we first recall the
quantity they use for an integer
$h\ge 2$: \( \operatorname{IS}_h(G) := \max\Bigl\{
|I|: I\subseteq V(G) \text{ contains no
clique of size }h \Bigr\}. \)

For $h=2$, a set contains no clique
of size $2$ exactly when it contains
no edge. Hence $\operatorname{IS}_2(G)=\operatorname{IS}(G)$, and their problem
reduces to the usual gap Independent Set
problem.

Their gap-$h$-Clique-Independent-Set problem, denoted $h\mathrm{IS}(r,\varepsilon,h)$, asks, given
a $(K,r)$-co-partite graph $G$, to distinguish between
\[
  \operatorname{IS}(G)=K
  \qquad\text{and}\qquad
  \operatorname{IS}_h(G)\le \varepsilon K.
\]

\citet[Theorem~2.1]{dinursafra2005vertexcover} state that, for every $h,\varepsilon>0$, the
problem $h\mathrm{IS}(r,\varepsilon,h)$ is NP-hard whenever $r\ge (h/\varepsilon)^c$
for an absolute constant $c>0$. We specialize
this theorem to $h=2$ and $\varepsilon=1/128$. Choose
once and for all an integer $r\ge
256^c$. Since $h$ and $\varepsilon$ are fixed
constants, so is $r$. Theorem~2.1 then
implies that it is NP-hard, given
a $(K,r)$-co-partite graph $G$, to distinguish
\[
  \operatorname{IS}(G)=K
  \qquad\text{from}\qquad
  \operatorname{IS}(G)\le \frac{K}{128}.
\]
Finally, since $N=Kr$, or equivalently $K=N/r$, this
is exactly the gap stated in~\eqref{eq:source-gap}.

\subsection{Proof of Proposition~\ref{prop:complete}}
\label{app:proof-complete}

The schedule analyzed below is summarized by
\begin{tcolorbox}[
  enhanced, sharp corners, colback=brightblue!4, colframe=brightblue!65,
  boxrule=0.5pt, boxsep=0pt, left=4pt, right=4pt, top=6pt, bottom=6pt,
  oversize, before skip=7pt plus 2pt, after skip=7pt plus 2pt,
  before upper={
    \setlength{\abovedisplayskip}{0pt}\setlength{\belowdisplayskip}{0pt}
    \setlength{\abovedisplayshortskip}{0pt}\setlength{\belowdisplayshortskip}{0pt}
  }
]
\[
\underbrace{\text{setup}}_{\text{\small $S_t(b_v):0\to -k$}}
\hspace{0.6em}\longrightarrow\hspace{0.6em}
\underbrace{\text{activate }v}_{\text{\small $S_t(b_v):-k\to k/2$}}
\hspace{0.6em}\longrightarrow\hspace{0.6em}
\underbrace{\text{harvest }v}_{\text{\small $S_t(b_v):k/2\to 0,\;\text{earn}\approx k/2$}}
\hspace{0.6em}\longrightarrow\hspace{0.6em}
\text{next }v\ldots
\]
\end{tcolorbox}
\begin{proofof}{Proposition~\ref{prop:complete}.}
Fix an arbitrary ordering of the vertices
in $I$. We construct the pure trajectory
in three stages: first a setup phase,
then an activation--harvesting block for each vertex
in $I$, and finally padding if necessary.

First, play $p$ for $q_0=Nk$ rounds. At
the end of the setup phase, $X_t(p)/N=k$.
This ratio remains equal to $k$ throughout
the activation--harvesting blocks, since no further setup
actions are played until the final padding.
No activation or harvesting action has been
played yet, so \eqref{eq:score-identities} gives $S_t(b_v)=-k$ and
$S_t(d_v)=-3k/2$ for every vertex $v$, and
$S_t(d_{uv})=-3k/2$ for every edge $uv\in\mathcal{E}$. Thus
the setup phase puts all rewarding
actions below the neutral action $b_0$.

Now process the vertices of $I$ one
at a time. For each $v\in I$,
play $a_v$ for $3k/2$ rounds, followed by
$h_v$ for $k/2$ rounds. Since $k$ is
even, both block lengths are integers.

Consider the moment immediately after the activation
block of $v$ and before its harvesting
block begins. At this point, $X_t(a_v)=3k/2$,
$X_t(h_v)=0$, and $X_t(p)/N=k$. Hence the rewarding
action associated with $v$ has score
\[
  S_t(b_v)
  =X_t(a_v)-X_t(h_v)-\frac{X_t(p)}{N}
  =\frac{3k}{2}-k
  =\frac{k}{2}.
\]
Thus activation raises the score of $b_v$
from $-k$ to $k/2$.

We next check the scores of all
competing learner actions during the harvesting block
of $v$. First, consider a vertex $u\in
I$ that was processed before $v$. Its
activation and harvesting blocks are already
complete, so $X_t(a_u)=3k/2$ and $X_t(h_u)=k/2$. Therefore,
\[
  S_t(b_u)
  =X_t(a_u)-X_t(h_u)-\frac{X_t(p)}{N}
  =\frac{3k}{2}-\frac{k}{2}-k
  =0.
\]
Thus, each previously harvested rewarding action has
returned to a score of exactly zero.
If a vertex $w$ has not yet
been processed, then $X_t(a_w)=X_t(h_w)=0$, and hence $S_t(b_w)=-k$.
The same conclusion holds for every vertex
outside $I$, since its activation and harvesting
actions are never played in this construction.

Now consider the local alarms. For the
current vertex $v$,
\[
  S_t(d_v)
  =X_t(a_v)-\frac{3X_t(p)}{2N}
  =\frac{3k}{2}-\frac{3k}{2}
  =0.
\]
A previously processed vertex $u\in I$
also has $X_t(a_u)=3k/2$, so its local
alarm likewise has a score of
zero. Every vertex that has not
been activated has local-alarm score $-3k/2$.
Consequently, $S_t(d_w)\le0$ for every vertex $w$.

Finally, consider an edge alarm $d_{uw}$. All
vertices activated so far belong to $I$.
Since $I$ is an independent set, an
edge $uw\in\mathcal{E}$ cannot have both endpoints
among these vertices. Therefore, at most
one of $X_t(a_u)$ and $X_t(a_w)$ equals
$3k/2$, while the other is zero.
In particular, $X_t(a_u)+X_t(a_w)\le3k/2$, and \eqref{eq:score-identities} gives
\[
  S_t(d_{uw})
  =X_t(a_u)+X_t(a_w)-\frac{3X_t(p)}{2N}
  \le\frac{3k}{2}-\frac{3k}{2}
  =0.
\]
Thus, at the beginning of the
harvesting block of $v$, we have
$S_t(b_v)=k/2$, while every other learner action,
including the neutral action $b_0$, has
a score of at most zero.

During the harvesting block, the optimizer repeatedly
plays $h_v$. Suppose that $s$ harvesting rounds
for $v$ have already been completed,
where $0\le s<k/2$. Then $X_t(h_v)=s$, while
$X_t(a_v)=3k/2$ and $X_t(p)/N=k$ remain unchanged. Therefore,
\[
  S_t(b_v)=\frac{3k}{2}-s-k=\frac{k}{2}-s.
\]
Hence the score of $b_v$ starts at
$k/2$ and decreases by exactly one after
each harvesting round. After all $k/2$
harvesting rounds, it returns to zero.

Harvesting changes neither the activation counts
nor $X_t(p)$, so all local and
edge alarm scores remain unchanged. The
other rewarding scores also remain unchanged.
Therefore, throughout the harvesting block, $S_t(b)\le0$
for every $b\ne b_v$, whereas $S_t(b_v)=k/2-s$.

Now consider a harvesting round for which
$S_t(b_v)>\gamma_TT$. Since every other learner action has
a score of at most zero, $S_t(b_v)>S_t(b)+\gamma_TT$
for every $b\ne b_v$. By $\gamma_T$-mean-basedness,
$y_t(b)\le\gamma_T$ for every $b\ne b_v$. Since
the learner's probabilities sum to one,
\[
  y_t(b_v)
  =1-\sum_{b\ne b_v}y_t(b)
  \ge1-(m-1)\gamma_T.
\]
Because the optimizer is playing $h_v$ and
$A(h_v,b_v)=1$ is the only nonzero payoff in
this row, its reward in this round
is exactly $r_t=e_{h_v}^{\top}Ay_t=y_t(b_v)$. Thus every harvesting round
satisfying $S_t(b_v)>\gamma_TT$ earns at least $1-(m-1)\gamma_T$.

This argument may fail only near the
end of the block, when $S_t(b_v)\le\gamma_TT$. Since
$S_t(b_v)$ decreases by exactly one per harvesting
round, there are at most $\gamma_TT+1$
such rounds for each vertex $v$.
Across all vertices in $I$, discarding
these rounds loses at most $|I|(\gamma_TT+1)$
relative to the ideal reward $|I|k/2$.

On every remaining harvesting round, the probability
mass assigned to actions other than \(b_v\)
(the \emph{probability leakage} away from \(b_v\)) satisfies
\[
1-y_t(b_v)\le (|\mathcal B|-1)\gamma_T.
\]
Since there are at most $T$ rounds
in total, the total leakage loss is
at most $m\gamma_TT$. Combining the two losses
gives
\[
  R(\pi,T)
  \ge |I|\frac{k}{2}
    -|I|(\gamma_TT+1)
    -m\gamma_TT.
\]

It remains only to verify that the
trajectory fits within the horizon. Each vertex
$v\in I$ uses $3k/2+k/2=2k$ rounds, so all
activation--harvesting blocks use $2k|I|\le2Nk$ rounds. Together with
the $q_0=Nk$ setup rounds, the construction uses
at most $Nk+2Nk=3Nk=T$ rounds. Fill any unused
rounds with $p$. Since the optimizer receives
zero reward when playing $p$, this
padding does not reduce the reward
already collected. This proves the claim.
\end{proofof}

\subsection{Proofs of Lemma~\ref{lem:independent} and Proposition~\ref{prop:sound}}
\label{app:proof-sound}

We first show that each extracted set
$I_j$ is independent, and then bound the
reward of an arbitrary mixed trajectory by
the size of the largest extracted set
$I_{\mathrm{alg}}$.

\begin{proofof}{Lemma~\ref{lem:independent}.}
Fix $j$, and suppose for contradiction that
$I_j$ contains two adjacent vertices $u$ and
$v$. By the definition of $I_j$, there
exist times $t$ and $t'$ such that
$v\in J_t$ and $u\in J_{t'}$, and both
setup levels lie in the same interval:
\[
  L_j\le\frac{X_t(p)}{N}<\frac32L_j,
  \qquad
  L_j\le\frac{X_{t'}(p)}{N}<\frac32L_j.
\]
Without loss of generality, assume $t\le t'$.

Since $v\in J_t$, the definition of $J_t$
gives $X_t(a_v)\ge9X_t(p)/(10N)$. Because $X_t(p)/N\ge L_j$, we obtain
$X_t(a_v)\ge9L_j/10$. The quantity $X_t(a_v)$ is the cumulative
amount of activation action $a_v$ played before
time $t$, and therefore it is nondecreasing
in time. Hence, since $t\le t'$,
\begin{equation*}
  X_{t'}(a_v)\ge X_t(a_v)\ge\frac9{10}L_j.
  \tag{$*$}
\end{equation*}

Now use the fact that $u\in J_{t'}$.
Since $u$ and $v$ are adjacent, we
have $v\in N_G(u)$. Therefore, the neighbor condition
in the definition of $J_{t'}$ requires
\[
  X_{t'}(a_v)\le\frac{3X_{t'}(p)}{5N}.
\]
But $t'$ belongs to the same setup
interval, so $X_{t'}(p)/N<3L_j/2$. Consequently,
\begin{equation*}
  X_{t'}(a_v)
  <\frac35\cdot\frac32L_j
  =\frac9{10}L_j.
  \tag{$**$}
\end{equation*}
The bounds $(*)$ and $(**)$ contradict each
other. Hence no two adjacent vertices
can both belong to $I_j$, and
therefore $I_j$ is an independent set.

The same argument also covers $t=t'$:
membership of $v$ in $J_t$ forces
$X_t(a_v)\ge9L_j/10$, while membership of its neighbor
$u$ in the same $J_t$ would
force $X_t(a_v)<9L_j/10$. This proves the claim.
\end{proofof}

\begin{proofof}{Proposition~\ref{prop:sound}.}
The optimizer receives reward only from harvesting
actions. In round \(t\), the reward identity
\eqref{eq:reward-identity} gives
$$ r_t=\sum_{v\in V}x_t(h_v)y_t(b_v). $$
We bound this reward by separating the
contribution of vertices that are not
near-leading from those that are near-leading.

First, consider a vertex \(v\) that is
not near-leading at time \(t\). Failure of
the near-leading condition means that some learner
action has a score exceeding \(S_t(b_v)\) by
more than \(\gamma_TT\). Since the learner is
\(\gamma_T\)-mean-based, we have \(y_t(b_v)\le\gamma_T\). Therefore, the total
reward contributed in round \(t\) by vertices
that are not near-leading is at most
$$ \sum_{\{v:\,v\text{ is not near-leading at }t\}} x_t(h_v)y_t(b_v) \le \gamma_T\sum_v x_t(h_v) \le \gamma_T. $$
Summing over all $T$ rounds gives a
total contribution of at most $\gamma_TT$. Hence
\begin{align}
  R(x,T) \le \gamma_TT + \sum_{t=0}^{T-1} \sum_{\{v:\,v\text{ is near-leading at }t\}} x_t(h_v).
  \label{eq:upper-bound-total-reward}
\end{align}
This reduces the problem to bounding how
much harvesting mass can be assigned to
near-leading vertices.

In the proof, we will repeatedly use
the following simple observation. Fix a vertex
\(v\), and consider any collection of rounds
on which $X_t(h_v)\le U$ for some \(U\).
Then the total harvesting mass of \(v\)
over those rounds is at most $U+1$.
To see this, focus on the last
selected round and call its index \(t^\star\).
All harvesting mass assigned to \(v\) on
earlier selected rounds is already included
in \(X_{t^\star}(h_v)\), while the current round
contributes $x_{t^\star}(h_v)\le 1$. Thus the total
selected harvesting mass is at most
$$ X_{t^\star}(h_v)+x_{t^\star}(h_v)\le U+1. $$
We next separate the near-leading rounds
according to the current setup level.

\paragraph{Small setup levels}
First, consider near-leading rounds satisfying
$$ \frac{X_t(p)}{N}<10(\gamma_TT+1). $$
For a near-leading vertex \(v\), the~\ref{eq:local-comparisons} give
\begin{align}
  X_t(h_v) \le \frac{X_t(p)}{2N}+\gamma_TT.
  \label{eq:local-harvesting-bound}
\end{align}
Therefore, throughout this range,
$$ X_t(h_v) < 5(\gamma_TT+1)+\gamma_TT = 6(\gamma_TT+1)-1 < 6(\gamma_TT+1).$$

Applying the preceding observation, the total harvesting
mass of a fixed vertex \(v\) over
all such rounds is at most $6(\gamma_TT+1).$
There are \(N\) vertices, so the total
contribution from all near-leading rounds with a
small setup level is at most
$6N(\gamma_TT+1).$ We may therefore treat the
contribution from rounds with small setup
levels as an additive error term.

\paragraph{Large setup levels}
It remains to consider near-leading rounds for
which
$$ \frac{X_t(p)}{N}\ge 10(\gamma_TT+1). $$
In particular, $\gamma_TT\le \frac{X_t(p)}{10N}.$ If \(v\) is
near-leading, the~\ref{eq:local-comparisons} give
$$ X_t(a_v) \ge \frac{X_t(p)}{N}+X_t(h_v)-\gamma_TT \ge \frac{X_t(p)}{N}-\gamma_TT \ge \frac{9X_t(p)}{10N}. $$
For every neighbor \(u\in N_G(v)\),
$$ X_t(a_u) \le \frac{X_t(p)}{2N}+\gamma_TT-X_t(h_v) \le \frac{X_t(p)}{2N}+\gamma_TT \le \frac{3X_t(p)}{5N}. $$
These are precisely the conditions defining \(J_t\).
Hence, at every sufficiently large setup level,
\begin{align*}
  v\text{ near-leading at time }t \quad\Longrightarrow\quad v\in J_t.
\end{align*}
Thus every vertex that can contribute to
the remaining reward is captured by the
candidate set \(J_t\).

Now fix one geometric setup interval $[L_j,3L_j/2).$
Recall that \(I_j\) collects the candidates from
all sets \(J_t\) that appear while
$$ L_j\le \frac{X_t(p)}{N}<\frac32L_j. $$
Lemma~\ref{lem:independent} shows that \(I_j\) is an independent
set. Therefore, throughout this interval, every near-leading
vertex that can contribute to the reward
belongs to the same independent set \(I_j\).

We next bound how much a single
vertex can contribute inside this interval.
Let \(v\in I_j\) be near-leading at
some round in the interval. The
local-alarm inequality in \eqref{eq:local-harvesting-bound} again gives
$$ X_t(h_v) \le \frac{X_t(p)}{2N}+\gamma_TT. $$
In the large-setup regime, $\gamma_TT\le \frac{X_t(p)}{10N}, $
and hence
$$ X_t(h_v) \le \frac{3X_t(p)}{5N}. $$
Because the setup level lies in \([L_j,3L_j/2)\),
$$ X_t(h_v) < \frac35\cdot\frac32L_j = \frac9{10}L_j. $$
Applying the earlier observation with \(U=9L_j/10\),
the total harvesting mass contributed by
\(v\) while the setup level remains
in this interval is at most
$$ \frac9{10}L_j+1<2L_j, $$
where we used \(L_j\ge1\).

Therefore, the total contribution from the \(j\)-th
setup interval is at most $ 2L_j|I_j|.
$ By construction, \(I_{\mathrm{alg}}\) is a largest
extracted set, so $|I_j|\le |I_{\mathrm{alg}}| $ for
every \(j\). Hence the contribution from interval
\(j\) is at most $2L_j|I_{\mathrm{alg}}|. $ It
remains only to sum over the geometric
intervals. Since $X_t(p)\le T, $ the
setup level never exceeds $ \frac{T}{N}.
$ Thus, if \(L_J\) is the
largest lower endpoint of a relevant
interval $[L_J,3L_J/2)$, then $L_J\le \frac{T}{N},$ and
$$ \sum_{j=0}^J L_j = \sum_{j=0}^J\left(\frac32\right)^j = \frac{(3/2)^{J+1}-1}{(3/2)-1} = 2\left(\frac32\right)^{J+1}-2 = 3L_J-2. $$
Consequently, the total contribution from all
large-setup near-leading rounds is at most
$$ \sum_j2L_j|I_{\mathrm{alg}}| \le 2|I_{\mathrm{alg}}| \sum_jL_j < 6|I_{\mathrm{alg}}|\frac{T}{N}. $$
Combining the three contributions gives
\begin{align*}
  R(x,T) \le \underbrace{
6|I_{\mathrm{alg}}|\frac{T}{N}
}_{\substack{\text{controlled by the}\\ \text{independent set}}}
\;+\;
\underbrace{
6N(\gamma_T T+1)\vphantom{\frac{T}{N}}
}_{\substack{\text{small-setup error}}}
\;+\;
\underbrace{
\gamma_T T\vphantom{\frac{T}{N}}
}_{\substack{\text{not-near-leading}\\ \text{mean-based leakage}}},
\end{align*}
as claimed.

Finally, suppose the trajectory is explicitly represented
and rational. We can compute every prefix
count \(X_t(p)\) and \(X_t(a_v)\) exactly, test the
rational inequalities defining \(J_t\), place each round
into its corresponding geometric setup interval, form
the sets \(I_j\), and return one of
maximum cardinality. There are only \(O(1+\log T)\)
relevant intervals, and all interval endpoints and
prefix counts have polynomial encoding length. Hence,
the extraction runs in polynomial time. The
procedure uses only the optimizer trajectory: it
does not evaluate the learner, compute the
learner's reward, or require knowledge of \(\gamma_T\).
\end{proofof}

\subsection{Horizon choice, error bounds, and proof of
the main theorem}
\label{app:horizon-bookkeeping}

Recall from Section~\ref{sec:hardness-transfer} the choice of the
block scale $k$, the setup length $q_0$,
and the horizon $T$; we verify here
that this choice is computable in polynomial
time, that it produces the scales in
\eqref{eq:reduction-scales}, and that the error terms of
Propositions~\ref{prop:complete} and~\ref{prop:sound} are sublinear. Since \(2/\zeta-1\) is
a fixed rational number, \(\lceil N^{2/\zeta-1}\rceil\) can
be computed exactly in polynomial time. The
horizon is polynomial in the source
size, and each payoff matrix has
$O(N^3)$ entries of $O(\log N)$ bits.
Both action counts are $O(N^2)=O(T^\zeta)$ and
hence at most $T^\beta$ for sufficiently
large instances. The assumption $\gamma_T=o(T^{-\zeta})$ implies
\begin{equation}
  N(\gamma_TT+1)=o(T),
  \qquad m\gamma_TT=o(T).
  \label{eq:reduction-errors}
\end{equation}
Consequently, Propositions~\ref{prop:complete} and~\ref{prop:sound} give
\[
  R_{\mathrm{pure}}^\star(T)\ge\frac{\operatorname{IS}(G)}{6N}T-o(T),
  \qquad
  R(x,T)\le6\frac{|I_{\mathrm{alg}}|}{N}T+o(T),
\]
uniformly over graph inputs and optimizer trajectories.
Now we are ready to prove the
main theorem.

\begin{proofof}{Theorem~\ref{thm:main}.}
We now choose the hardness constant. With
$r$ from~\eqref{eq:source-gap}, we take $c=1/(24r)$, which depends
only on the source gap. Suppose a
polynomial-time optimizer guarantees the claimed additive approximation.
Given a source graph, construct the game
of Section~\ref{sec:construction} with the horizon of Section~\ref{sec:hardness-transfer}.
If the optimizer takes a mean-based parameter bound $\bar\gamma_T$
as additional input, supply $\bar\gamma_T=1/\lceil T^\zeta\rceil$ as
in Remark~\ref{rem:computational-model}; this is valid for all
sufficiently large reduction horizons and computable in
polynomial time. The reward errors above still
use the learner's intrinsic parameter $\gamma_T$. Run
the optimizer and extract $I_{\mathrm{alg}}$ using Proposition~\ref{prop:sound}.
Completeness, the hypothetical guarantee, and soundness imply
\[
  \frac{\operatorname{IS}(G)}{6N}T-cT-o(T)
    \le R_{\mathrm{alg}}(T)
    \le6\frac{|I_{\mathrm{alg}}|}{N}T+o(T).
\]
Rearranging gives
\[
  |I_{\mathrm{alg}}|\ge\frac{\operatorname{IS}(G)}{36}-\frac c6 N-o(N).
\]
On a YES instance of \eqref{eq:source-gap}, this
yields
\[
  |I_{\mathrm{alg}}|\ge\frac{N}{48r}-o(N)>\frac{N}{64r}
\]
for all sufficiently large instances. On a
NO instance, every independent set has size
at most $N/(128r)$. Comparing the extracted size
with the rational threshold $N/(64r)$ therefore decides
the NP-hard promise problem in polynomial time;
the finitely many smaller source sizes
can be handled separately. No step
evaluates the learner or the trajectory
reward. This contradiction proves the theorem.
\end{proofof}

% !TEX root = ../main.tex
\section{Setwise mean-based learners}
\label{app:learner-classes}

The hardness theorem uses a single parameter
to control both the score gap that
makes an action inferior and the probability
assigned to that action. For separable \ac{FTRL},
we can instead bound the total probability
of an inferior set directly. This separates
the score threshold from the probability error
and makes the learning-rate conditions easier to
derive. We first show how the
reduction uses this setwise guarantee, then
obtain it from the regularizer's derivative.

\subsection{What the reduction needs: setwise mean-based learners}
\label{app:setwise-extension}

The reduction needs little mass on
inferior actions, regardless of how that
mass is distributed within the inferior
set. To state this property, call
a learner $(\gamma_T,\rho_T)$-\emph{setwise mean-based} if, for
every $\mathcal S\subseteq\mathcal{B}$ and $b^\star\notin\mathcal S$,
\[
  S_t(b^\star)-S_t(b)>\gamma_TT
  \quad\text{for every }b\in\mathcal S
  \qquad\Longrightarrow\qquad
  y_t(\mathcal S):=\sum_{b\in\mathcal S}y_t(b)\le\rho_T.
\]
Here $\gamma_TT$ is the required score gap,
and $\rho_T$ bounds the total probability that
may remain on inferior actions. No
other restriction is placed on the
learner's dependence on the ordered payoff
history or its fixed internal memory.

\begin{corollary}[Setwise mean-based learners]
\label{cor:setwise-hardness}
The conclusion of Theorem~\ref{thm:main} also holds for
any fixed deterministic $(\gamma_T,\rho_T)$-setwise mean-based learner family
with
\[
  \gamma_TT=o(T^{1-\zeta}),
  \qquad \rho_T=o(1),
\]
using the same game family and the
same absolute constant $c>0$.
\end{corollary}

\begin{proof}
The score comparisons and independent-set extraction remain
unchanged. We only replace the probability error
in the completeness and soundness reward
bounds by the setwise bound $\rho_T$.

During each profitable harvesting round of Proposition~\ref{prop:complete},
the rewarding action is more than $\gamma_TT$
ahead of every other action. The setwise mean-based condition
therefore gives $y_t(b_v)\ge1-\rho_T$. In the soundness argument
of Proposition~\ref{prop:sound}, a highest-scoring action is more
than $\gamma_TT$ ahead of every rewarding action
that is not near-leading. Their total
probability, and hence their total reward
contribution, is at most $\rho_T$ in
each round. The resulting bounds are
\begin{align*}
  R(\pi,T)&\ge |I|\frac{k}{2}-|I|(\gamma_TT+1)-\rho_TT,\\
  R(x,T)&\le6|I_{\mathrm{alg}}|\frac{T}{N}+6N(\gamma_TT+1)+\rho_TT.
\end{align*}
Both error terms are sublinear at the
reduction horizons. Indeed, the choice in
\eqref{eq:reduction-scales} gives $N(\gamma_TT+1)=o(T)$, while $\rho_TT=o(T)$ by
assumption. The hardness transfer in Section~\ref{sec:hardness-transfer}
therefore applies with the same $c=1/(24r)$.
\end{proof}

\paragraph{Coordinatewise mean-based learners.}
The original mean-based assumption implies a setwise
guarantee by summing the coordinatewise probability bounds.
Suppose the learner is coordinatewise $\bar\gamma_T$-mean-based: whenever
\[
  S_t(b^\star)-S_t(b)>\bar\gamma_TT,
\]
it assigns $y_t(b)\le\bar\gamma_T$. If one action
$b^\star\notin\mathcal S$ is more than $\bar\gamma_TT$
ahead of every member of $\mathcal
S\subseteq\mathcal{B}$, summing over the set gives
\begin{equation}
  y_t(\mathcal S)
  =\sum_{b\in\mathcal S}y_t(b)
  \le |\mathcal S|\bar\gamma_T
  \le m\bar\gamma_T.
  \label{eq:app-coordinatewise-to-setwise}
\end{equation}
Thus a coordinatewise $\bar\gamma_T$-mean-based learner is $(\bar\gamma_T,\rho_T)$-setwise
mean-based with
\begin{equation*}
  \rho_T:=m\bar\gamma_T.
\end{equation*}
For the constructed family, $m=O(T^\zeta)$, so $\bar\gamma_TT=o(T^{1-\zeta})$
already implies $m\bar\gamma_T=o(1)$. Thus the main
theorem's mean-based assumption also supplies vanishing
setwise error along the reduction horizons.

\subsection{Setwise guarantees for separable \acs{FTRL}}
We now derive a setwise guarantee for
the separable \ac{FTRL} model from Section~\ref{sec:ftrl-preliminaries}. If
an inferior set carries too much mass,
transferring mass from one of its coordinates
to a higher-scoring action improves the linear
part of the \ac{FTRL} objective. A sufficiently
large score gap makes this improvement
exceed the regularization cost, contradicting optimality.

The derivative of the regularizer determines the
score gap needed for this comparison. For
$0<\rho<1/2$, we express this threshold as
\begin{equation}
  G_h(\rho,m;\eta_T)
  :=\frac{
      \theta'(1-\rho)-\theta'(\rho/m)
    }{\eta_T}.
  \label{eq:app-ftrl-gap}
\end{equation}
The next proposition shows that this threshold
forces the total mass on any inferior
set to be at most $\rho$.

\begin{proposition}[Setwise mean-based guarantees for separable \acs{FTRL}]
\label{prop:app-ftrl-setwise}
Fix $0<\rho<1/2$. Let $\mathcal S\subseteq\mathcal{B}$ and $b^\star\notin\mathcal
S$. If
\begin{equation}
  S_t(b^\star)-S_t(b)
  >G_h(\rho,m;\eta_T)
  \qquad\text{for every }b\in\mathcal S,
  \label{eq:app-ftrl-gap-premise}
\end{equation}
then the \ac{FTRL} response satisfies
\begin{equation*}
  y_t(\mathcal S)\le\rho.
\end{equation*}
\end{proposition}

\begin{proof}
Suppose, toward a contradiction, that $y_t(\mathcal S)>\rho$,
and let $r:=|\mathcal S|$. We first identify
a coordinate from which to transfer mass.
At least one $b\in\mathcal S$ must satisfy
\begin{equation}
  y_t(b)>\frac{\rho}{r}\ge\frac{\rho}{m},
  \label{eq:app-ftrl-large-bad-coordinate}
\end{equation}
while the fact that $b^\star\notin\mathcal S$ gives
\begin{equation}
  y_t(b^\star)<1-\rho.
  \label{eq:app-ftrl-benchmark-mass}
\end{equation}
Since $y_t(b)>0$, moving a sufficiently small amount
of mass from $b$ to $b^\star$
is feasible. Optimality in \eqref{eq:app-ftrl-response} requires
the resulting linear gain to be
at most the marginal regularization cost.
The one-sided directional derivative therefore gives
\begin{equation*}
  \eta_T\bigl(S_t(b^\star)-S_t(b)\bigr)
  \le
  \theta'\bigl(y_t(b^\star)\bigr)
  -\theta'\bigl(y_t(b)\bigr).
\end{equation*}
Equivalently, compare the objective at $y_t$
and $y_t+\delta(e_{b^\star}-e_b)$ and let $\delta\downarrow0$. This
also handles finite endpoint derivatives. The
coordinate bounds in \eqref{eq:app-ftrl-large-bad-coordinate} and~\eqref{eq:app-ftrl-benchmark-mass}, together
with the increasing derivative $\theta'$, imply
\begin{equation*}
\begin{aligned}
  \eta_T\bigl(S_t(b^\star)-S_t(b)\bigr)
  &<\theta'(1-\rho)-\theta'(\rho/m).
\end{aligned}
\end{equation*}
This contradicts \eqref{eq:app-ftrl-gap-premise} and proves the claim.
\end{proof}

To apply the hardness theorem, we need
both a vanishing probability error and
a sufficiently small score threshold. Choosing
a sequence $\rho_T\downarrow0$ makes separable \ac{FTRL}
$(\gamma_T,\rho_T)$-setwise mean-based at the gap scale
\begin{equation*}
  \gamma_TT
  :=G_h(\rho_T,m;\eta_T).
\end{equation*}
Hence Corollary~\ref{cor:setwise-hardness} applies whenever
\begin{equation}
  G_h(\rho_T,m;\eta_T)=o(T^{1-\zeta}).
  \label{eq:app-ftrl-master-rate}
\end{equation}
We now evaluate this condition for the
standard regularizers. The dependence of the derivative
gap on $m$ and $\rho_T$ determines how
large the learning rate must be.

\paragraph{Negative entropy and Euclidean regularization.}
To compare the regularizers under the same
learning-rate regime, we use the polynomial schedule
\begin{equation*}
  \eta_T=T^{-(1-\alpha)},
  \qquad
  \alpha\in(0,1],
\end{equation*}
and choose $\rho_T=T^{-\nu}$ for an arbitrary fixed
$\nu>0$. For negative entropy,
\[
  \theta(s)=s\log s,
  \qquad
  \theta'(s)=1+\log s,
\]
so
\begin{equation*}
  G_h(\rho_T,m;\eta_T)
  =\frac1{\eta_T}
   \log\!\left(\frac{(1-\rho_T)m}{\rho_T}\right).
\end{equation*}
The action count $m$ is polynomial
in $T$, so this gap is
$O(T^{1-\alpha}\log T)$. It satisfies \eqref{eq:app-ftrl-master-rate} whenever
\begin{equation*}
  \zeta<\alpha.
\end{equation*}
For the Euclidean regularizer
\[
  \theta(s)=\frac12s^2,
  \qquad
  \theta'(s)=s,
\]
we have
\begin{equation*}
  G_h(\rho_T,m;\eta_T)
  =\frac{(1-\rho_T)-\rho_T/m}{\eta_T}
  \le T^{1-\alpha},
\end{equation*}
so $\zeta<\alpha$ suffices for both negative entropy
and Euclidean \ac{FTRL} under this learning-rate schedule.

\paragraph{Negative Tsallis entropy.}
For a fixed $q\in(0,1)$, the convex negative-Tsallis
regularizer has
\begin{equation*}
  \theta_q(s):=-\frac{s^q}{q(1-q)},
  \qquad
  \theta_q'(s)=-\frac{s^{q-1}}{1-q}.
\end{equation*}
Then
\begin{equation*}
\begin{aligned}
  G_h(\rho_T,m;\eta_T)
  &=\frac{1}{(1-q)\eta_T}
    \left[
      \left(\frac{m}{\rho_T}\right)^{1-q}
      -(1-\rho_T)^{-(1-q)}
    \right]
  \notag\\
  &=O\!\left(
      \eta_T^{-1}
      \left(\frac{m}{\rho_T}\right)^{1-q}
    \right).
\end{aligned}
\end{equation*}
For negative Tsallis entropy, the gap grows
as a power of $m/\rho_T$. The action
count therefore enters the rate condition
polynomially. With $m\le T^\beta$ and $\rho_T=T^{-\nu}$,
\begin{equation*}
  G_h(\rho_T,m;\eta_T)
  =O\!\left(
    T^{1-\alpha+(\beta+\nu)(1-q)}
  \right).
\end{equation*}
The gap is $o(T^{1-\zeta})$ provided
\begin{equation}
  \zeta+(\beta+\nu)(1-q)<\alpha.
  \label{eq:app-tsallis-condition-nu}
\end{equation}
A sufficiently small $\nu>0$ satisfies this condition
whenever $\zeta+\beta(1-q)<\alpha$.

The constructed family allows a sharper
estimate. We can use the actual
action count $m=2N+|\mathcal{E}|+1=O(T^\zeta)$ instead of $m\le
T^\beta$. The sufficient condition then becomes
\begin{equation*}
  \zeta+\zeta(1-q)<\alpha,
  \qquad\text{equivalently}\qquad
  (2-q)\zeta<\alpha.
\end{equation*}

\paragraph{Extension to log-barrier regularization.}
The same mean-based argument also applies to the log barrier, which
lies outside Definition~\ref{def:separable-regularizer} because it is
infinite at the simplex boundary. For
\begin{equation*}
  \theta(s)=-\log s,
  \qquad
  \theta'(s)=-\frac1s,
\end{equation*}
the \ac{FTRL} response is interior, so the feasible-transfer argument
in Proposition~\ref{prop:app-ftrl-setwise} remains valid.
The derivative gives a gap linear in
$m/\rho_T$:
\begin{equation*}
\begin{aligned}
  G_h(\rho_T,m;\eta_T)
  &=\frac1{\eta_T}
    \left(
      \frac{m}{\rho_T}-\frac1{1-\rho_T}
    \right)
  =O\!\left(\frac{m}{\eta_T\rho_T}\right).
\end{aligned}
\end{equation*}
With $m\le T^\beta$ and $\rho_T=T^{-\nu}$,
\begin{equation*}
  G_h(\rho_T,m;\eta_T)
  =O\!\left(T^{1-\alpha+\beta+\nu}\right),
\end{equation*}
so choosing $\nu>0$ sufficiently small ensures \eqref{eq:app-ftrl-master-rate}
whenever $\zeta+\beta<\alpha$. For the actual action
count $m=O(T^\zeta)$, the sufficient condition improves
to $2\zeta<\alpha$. Thus the stronger divergence
of the log-barrier derivative near zero
imposes a more restrictive learning-rate condition
than negative entropy or Euclidean regularization.

% !TEX root = ../main.tex
\section{Omitted proofs for Section~\ref{app:zero-sum}}
\label{app:zero-sum-proofs}

Throughout this appendix, we use the zero-sum
model of Section~\ref{app:zero-sum}: $B=-A$, optimizer payoffs lie
in $[-1,1]$, the score update is $S_{t+1}=S_t-A^{\mathsf
T}x_t$ from $S_0=0$, and a pure trajectory
$\pi$ has $\pi_t=e_{\ell_t}$ for every $t$. From
Appendix~\ref{app:proof-passive} on, the learner uses \ac{FTRL}
with a separable regularizer $h$ as in Definition~\ref{def:separable-regularizer}.
The potential $\Phi_{\eta_T}$, the Bregman
divergence $D_{h^*}$, and the continuous-time reward $R_{\mathrm{cont}}(x^\star,T)$
are those of Section~\ref{sec:potential-decomposition}.
For the active results, we also impose the strong-convexity,
non-identification, and step-size assumptions of
Theorem~\ref{thm:app-pure-active}.

Appendix~\ref{app:proof-best-response} proves Proposition~\ref{prop:app-pure-best-response} and describes the polynomial-time
offline construction of the myopic trajectory. Appendix~\ref{app:proof-passive}
proves Lemma~\ref{lem:app-terminal-rounding}, Theorem~\ref{thm:app-pure-passive}, and Corollary~\ref{cor:pure-passive-inverse-rate}, and records
the terminal-rounding algorithm. Appendix~\ref{app:proof-active} proves Lemma~\ref{lem:app-cumulative-rounding} and
Theorem~\ref{thm:app-pure-active}, together with the derandomization, shadowing, and
one-step-bonus lemmas on which the proof rests,
and derives the asymptotic $\Omega(\eta_TT)$ surplus. Appendix~\ref{app:algorithmic}
explains how to carry out the construction
in polynomial time, and Appendix~\ref{app:random-games} applies
the active-surplus theorem to random games.

\subsection{Proof of Proposition~\ref{prop:app-pure-best-response} and the offline construction}
\label{app:proof-best-response}

A pure best response secures the game
value in each round, while the learner's
regret bounds the reward of every optimizer
trajectory. These two comparisons give the following
guarantee. We then explain how to
construct the trajectory by approximate simulation.

\begin{proposition}[Pure myopic best response]
\label{prop:app-pure-best-response}
At each round $t\in\{0,\ldots,T-1\}$, let $y_t$ be
the learner's response to the preceding optimizer
actions, and choose
\[
    \ell_t\in\operatorname*{arg\,max}_{\ell\in\mathcal{A}}e_\ell^{\mathsf T}Ay_t,
    \qquad \pi_t:=e_{\ell_t},
\]
with ties broken by the fixed action
order. The resulting pure trajectory $\pi_{\mathrm{BR}}=(e_{\ell_0},\ldots,e_{\ell_{T-1}})$ satisfies
\[
    R(\pi_{\mathrm{BR}},T)\ge T\operatorname{Val}(A)
\]
and
\[
    R(\pi_{\mathrm{BR}},T)
    \ge R_{\mathrm{mix}}^\star(T)
        -\operatorname{Regret}_{\mathsf L}(T)
    \ge R_{\mathrm{pure}}^\star(T)
        -\operatorname{Regret}_{\mathsf L}(T).
\]
\end{proposition}

\begin{proof}
We compare the trajectories through the game
value. For the myopic trajectory, the
minimax theorem gives, at every round,
\[
    \max_{\ell\in\mathcal{A}}e_\ell^{\mathsf T}Ay_t
    =\max_{u\in\Delta(\mathcal{A})}u^{\mathsf T}Ay_t
    \ge\min_{y\in\Delta(\mathcal{B})}
        \max_{u\in\Delta(\mathcal{A})}u^{\mathsf T}Ay
    =\operatorname{Val}(A).
\]
Thus the chosen pure action earns at
least $\operatorname{Val}(A)$, and summing over rounds gives
$R(\pi_{\mathrm{BR}},T)\ge T\operatorname{Val}(A)$.

For the upper bound, fix any optimizer
trajectory $x=(x_0,\ldots,x_{T-1})$, pure or mixed. The regret
bound gives
\begin{align*}
    R(x,T)
    &\le\min_{y\in\Delta(\mathcal{B})}
        \sum_{t=0}^{T-1}x_t^{\mathsf T}Ay
        +\operatorname{Regret}_{\mathsf L}(T)=T\min_{y\in\Delta(\mathcal{B})}
        \left(\frac1T\sum_{t=0}^{T-1}x_t\right)^{\mathsf T}Ay
        +\operatorname{Regret}_{\mathsf L}(T) \\
    &\le T\operatorname{Val}(A)
        +\operatorname{Regret}_{\mathsf L}(T),
\end{align*}
where the last inequality holds because the
average optimizer action is in $\Delta(\mathcal{A})$. Taking
the supremum over mixed trajectories yields
\[
    R_{\mathrm{mix}}^\star(T) = \sup_{x} R(x,T)
    \le T\operatorname{Val}(A)+\operatorname{Regret}_{\mathsf L}(T)
    \le R(\pi_{\mathrm{BR}},T)+\operatorname{Regret}_{\mathsf L}(T).
\]
Rearranging proves the mixed-benchmark guarantee. The
pure-benchmark bound follows from $R_{\mathrm{pure}}^\star(T)\le R_{\mathrm{mix}}^\star(T)$.
\end{proof}

\paragraph{Polynomial-time offline construction.}
Although Proposition~\ref{prop:app-pure-best-response} describes the optimizer sequentially, the
learner's response is determined by the preceding
history. Thus an efficiently simulatable learner allows
the optimizer to construct the entire trajectory
offline. Specifically, suppose that, given the rational
game description, the preceding pure optimizer actions,
and a rational precision parameter $\varepsilon\in(0,1)$, a
rational approximation to $y_t$ with $\ell_1$ error
at most $\varepsilon$ can be computed in
time polynomial in the input encoding length,
$T$, and $\log(1/\varepsilon)$. Starting from the learner's
initial state, the optimizer approximates its response,
selects a best row against that approximation,
and appends the chosen pure action to
the history. Repeating this procedure for $T$
rounds produces an explicit pure trajectory. Each
simulation uses the actual history of chosen
actions, so its error is measured against
the learner response induced by that history.

To quantify the loss from approximating $y_t$,
let $\pi=(e_{\ell_0},\ldots,e_{\ell_{T-1}})$ be the pure trajectory produced
by the approximate simulation. If each chosen
action satisfies
\[
    e_{\ell_t}^{\mathsf T}Ay_t
    \ge\max_{\ell\in\mathcal{A}}e_\ell^{\mathsf T}Ay_t-\varepsilon_t,
\]
then the same proof gives
\[
    R(\pi,T)
    \ge R_{\mathrm{mix}}^\star(T)
        -\operatorname{Regret}_{\mathsf L}(T)
        -\sum_{t<T}\varepsilon_t.
\]
Because the entries of $A$ lie in
$[-1,1]$, an $\ell_1$ approximation error of at
most $\varepsilon_t/2$ changes every row payoff by
at most $\varepsilon_t/2$. A best row against
the approximate response therefore satisfies the required
inequality. Comparing all rows takes $O(nm)$ arithmetic
operations per round. For example, choosing $\varepsilon_t=1/T$
gives total approximation loss at most
one and requires only polynomial-time simulation.
Consequently, if $\operatorname{Regret}_{\mathsf L}(T)=o(T)$, this procedure
achieves an $o(T)$ additive approximation to
the optimal mixed trajectory, and hence
also to the optimal pure trajectory.

\subsection{Proofs for passive extraction}
\label{app:proof-passive}

For passive extraction, it suffices to approximate
the terminal score of the fixed max-min
strategy $x^\star$ of Section~\ref{sec:potential-decomposition}. Intermediate prefixes need
not be controlled: the potential term depends
only on $S_T$, and the Bregman
terms in \eqref{eq:app-disc-cont-used} are nonnegative. The
following rounding lemma makes this approximation
explicit by keeping each row's total
number of plays between the floor
and ceiling of its target count.

\begin{lemma}[Terminally balanced pure implementation]
\label{lem:app-terminal-rounding}
For every $x^\star\in\Delta(\mathcal{A})$ and every horizon $T$,
there exist nonnegative integers $(q_\ell)_{\ell\in\mathcal{A}}$ satisfying
\[
    \sum_{\ell\in\mathcal{A}} q_\ell=T
\]
and
\[
    \lfloor Tx_\ell^\star\rfloor
    \le q_\ell
    \le
    \lceil Tx_\ell^\star\rceil
    \qquad
    \text{for every }\ell\in\mathcal{A}.
\]
Moreover, $q_\ell=0$ whenever $x_\ell^\star=0$. Consequently, there exists
a pure trajectory $\pi$ with $\pi_t=e_{\ell_t}$ in
which row $\ell$ appears exactly $q_\ell$ times,
and
\[
    \sum_{\ell\in\operatorname{supp}(x^\star)}
    \bigl|q_\ell-Tx_\ell^\star\bigr|
    \le
    |\operatorname{supp}(x^\star)|.
\]
The counts, and hence such a sequence,
can be computed in time polynomial in
$T$ and $n$.
\end{lemma}

\begin{proof}
We round each target count down, then
assign the remaining rounds to coordinates
with positive fractional parts. Initially, set
\[
    q_\ell:=\lfloor Tx_\ell^\star\rfloor
    \qquad
    \text{for every }\ell\in\mathcal{A}.
\]
The number of unassigned rounds is
\[
    r
    :=
    T-\sum_{\ell\in\mathcal{A}}\lfloor Tx_\ell^\star\rfloor
    =
    \sum_{\ell\in\mathcal{A}}
    \bigl(Tx_\ell^\star-\lfloor Tx_\ell^\star\rfloor\bigr).
\]
Each positive fractional part is less than
one, so there are enough such coordinates
to assign all $r$ remaining rounds. Increase
$q_\ell$ by one for any $r$ of
them. The resulting counts sum to $T$
and satisfy the stated floor--ceiling bounds.

If $x_\ell^\star=0$, then $Tx_\ell^\star=0$, so its floor
and ceiling are both zero and hence
$q_\ell=0$. Thus discrepancies occur only on $\operatorname{supp}(x^\star)$,
and every such discrepancy has absolute value
strictly less than one. This gives
\[
    \sum_{\ell\in\operatorname{supp}(x^\star)}
    |q_\ell-Tx_\ell^\star|
    \le
    |\operatorname{supp}(x^\star)|.
\]
Finally, arrange the multiset containing $q_\ell$
copies of row $\ell$ in any
order. Computing the counts and listing
these $T$ actions takes polynomial time.
\end{proof}

\begin{algorithm2e}[t]
\DontPrintSemicolon
\caption{Terminal rounding for passive extraction}
\label{alg:app-passive-terminal-greedy}
\KwIn{A target strategy $x^\star\in\Delta(\mathcal{A})$ and horizon $T$.}
\KwOut{A pure trajectory $\pi=(\pi_0,\ldots,\pi_{T-1})$.} \ForEach{$\ell\in\mathcal{A}$}{ Set $z_\ell\gets
Tx_\ell^\star$, $q_\ell\gets\lfloor z_\ell\rfloor$, and $f_\ell\gets z_\ell-\lfloor z_\ell\rfloor$\;
}
Set $r\gets T-\sum_{\ell\in\mathcal{A}}q_\ell$\; Let $I$ contain
the $r$ indices with largest positive
$f_\ell$, breaking ties by the fixed
action order\; \ForEach{$\ell\in I$}{Increase $q_\ell$ by
one\;} \For{$t=0,\ldots,T-1$}{ Using the counts $X_t(\ell)$
of the actions already chosen, select
  \[
    \ell_t\in\operatorname*{arg\,max}_{\ell:\,X_t(\ell)<q_\ell}
       \left\{\frac{(t+1)q_\ell}{T}-X_t(\ell)\right\},
  \]
  breaking ties by the fixed action order\;
  Set $\pi_t\gets e_{\ell_t}$\;
}
\Return{$\pi$}\;
\end{algorithm2e}

Algorithm~\ref{alg:app-passive-terminal-greedy} gives one deterministic ordering of the
rounded counts; Theorem~\ref{thm:app-pure-passive} shows that the passive
guarantee holds for every ordering, since only
the terminal score enters the potential comparison.

To track the passive surplus retained by rounding, we write
\[
    \operatorname{LAG}_{x^\star}(T)
    :=R_{\mathrm{cont}}(x^\star,T)-T\operatorname{Val}(A).
\]
This quantity is nonnegative because $x^\star$ is a max-min strategy.

\begin{proofof}{Theorem~\ref{thm:app-pure-passive}.}
Rounding changes the terminal score by at
most the support size. Since the potential
is Lipschitz and every Bregman term is
nonnegative, this score error bounds the loss
relative to the fixed strategy's continuous-time reward.
We make these two comparisons in turn.

Let $(q_\ell)_{\ell\in\mathcal{A}}$ be the counts from Lemma~\ref{lem:app-terminal-rounding},
and choose any ordering of these actions.
To express the resulting score error through
these counts, we record their terminal discrepancies
as
\[
    d_\ell:=q_\ell-Tx_\ell^\star.
\]
Rows outside the target support contribute no
discrepancy:
\[
    d_\ell=0
    \qquad
    \text{for }\ell\notin\operatorname{supp}(x^\star),
\]
and
\begin{equation}
\label{eq:app-terminal-discrepancy-l1}
    \sum_{\ell\in\operatorname{supp}(x^\star)}|d_\ell|
    \le
    |\operatorname{supp}(x^\star)|.
\end{equation}

The learner's terminal score under the pure
sequence is
\[
    S_T
    =
    -\sum_{\ell\in\mathcal{A}} q_\ell A^{\mathsf T} e_\ell
    =
    -TA^{\mathsf T} x^\star
    -
    \sum_{\ell\in\mathcal{A}} d_\ell A^{\mathsf T} e_\ell.
\]
The potential $\Phi_{\eta_T}$ is $1$-Lipschitz in
the $\ell_\infty$-norm, so a small terminal-score
error causes only a small loss
in its contribution to reward. To
see this, for arbitrary $U,V\in\mathbb R^m$,
\begin{align*}
    \Phi_{\eta_T}(U)-\Phi_{\eta_T}(V)
    &=
    \frac{1}{\eta_T}
    \left(
        h^*(\eta_T U)-h^*(\eta_T V)
    \right) \\
    &\le
    \max_{y\in\Delta(\mathcal{B})}\langle U-V,y\rangle
    \le
    \|U-V\|_\infty.
\end{align*}
Interchanging $U$ and $V$ gives
\[
    |\Phi_{\eta_T}(U)-\Phi_{\eta_T}(V)|
    \le
    \|U-V\|_\infty.
\]

It remains to bound the terminal-score
error to which this Lipschitz estimate
applies. Because $A\in[-1,1]^{n\times m}$, we have
\begin{align*}
    \left\|
        S_T+TA^{\mathsf T} x^\star
    \right\|_\infty
    &=
    \left\|
        \sum_{\ell\in\mathcal{A}} d_\ell A^{\mathsf T} e_\ell
    \right\|_\infty \\
    &\le
    \sum_{\ell\in\operatorname{supp}(x^\star)}|d_\ell| \\
    &\le
    |\operatorname{supp}(x^\star)|.
\end{align*}
Therefore
\begin{equation}
\label{eq:app-passive-terminal-potential}
    \Phi_{\eta_T}(0)-\Phi_{\eta_T}\bigl(S_T\bigr)
    \ge
    \Phi_{\eta_T}(0)
    -
    \Phi_{\eta_T}\bigl(-TA^{\mathsf T} x^\star\bigr)
    -
    |\operatorname{supp}(x^\star)|.
\end{equation}

Dropping the nonnegative Bregman terms from \eqref{eq:app-disc-cont-used}
gives
\[
    R(\pi,T)
    \ge
    \Phi_{\eta_T}(0)-\Phi_{\eta_T}\bigl(S_T\bigr).
\]
Combining the terminal-potential bound with \eqref{eq:app-zero-sum-continuous-reward} yields
\begin{align*}
    R(\pi,T)
    &\ge
    \Phi_{\eta_T}(0)
    -
    \Phi_{\eta_T}\bigl(-TA^{\mathsf T} x^\star\bigr)
    -
    |\operatorname{supp}(x^\star)| \\
    &=
    R_{\mathrm{cont}}(x^\star,T)
    -
    |\operatorname{supp}(x^\star)| \\
    &=
    T\operatorname{Val}(A)
    +
    \operatorname{LAG}_{x^\star}(T)
    -
    |\operatorname{supp}(x^\star)|.
\end{align*}
This proves \eqref{eq:app-pure-passive-main}.
\end{proofof}

The bounds of \citet{su2026nocoin} quantify the surplus
accumulated while the learner eliminates actions outside
$\operatorname{br}(x^\star)$. Terminal rounding transfers these bounds to
pure play with only the fixed support-size
loss.

\begin{proofof}{Corollary~\ref{cor:pure-passive-inverse-rate}.}
For $N_{\mathrm{sub}}>0$ and sufficiently large $\eta_TT$,
Proposition~\ref{prop:app-passive-ftrl-baseline} gives the continuous-time
surplus bound. Theorem~\ref{thm:app-pure-passive} transfers this bound to
a terminally balanced pure trajectory:
\begin{align*}
    R(\pi,T)
    &\ge R_{\mathrm{cont}}(x^\star,T)
          -|\operatorname{supp}(x^\star)| \\
    &\ge T\operatorname{Val}(A)
          +c_{A,h}\frac{N_{\mathrm{sub}}}{\eta_T}
          -|\operatorname{supp}(x^\star)|.
\end{align*}
Thus terminal rounding preserves the inverse-rate surplus up to the
fixed support-size loss. Algorithm~\ref{alg:app-passive-terminal-greedy}
constructs the required trajectory in polynomial time.
\end{proofof}

For fixed $A$, $x^\star$, and $h$, if $N_{\mathrm{sub}}>0$,
$\eta_T\to0$, and $\eta_TT\to\infty$, the support-size loss is fixed
while $N_{\mathrm{sub}}/\eta_T$ diverges. It can therefore be absorbed
into the inverse-rate term in Corollary~\ref{cor:pure-passive-inverse-rate},
giving
\[
    R(\pi,T)
    \ge T\operatorname{Val}(A)
        +\Omega_{A,h}\left(\frac{N_{\mathrm{sub}}}{\eta_T}\right).
\]

\subsection{Proofs for active extraction}
\label{app:proof-active}

This subsection proves Lemma~\ref{lem:app-cumulative-rounding} and Theorem~\ref{thm:app-pure-active} using
the ingredients developed in
Section~\ref{app:zero-sum-active}: we first construct a pure sequence
with uniformly bounded prefix discrepancies and show
how to derandomize it, then use the
resulting control of the learner's probabilities to
lower-bound each Bregman term, and finally combine
the terminal potential with the accumulated bonus.

\paragraph{The step-size condition.}
The one-update bound requires the distinguished best responses to
retain probability at least $\delta$ throughout each update.
On the reference ray $-tA^{\mathsf T}x^\star$, each best response
has probability at least $1/m$. To compare it with a pure
trajectory, recall
$L_A:=\max_{\ell\in\operatorname{supp}(x^\star)}\|A^{\mathsf T}e_\ell\|_*$
and let $c_{\|\cdot\|}:=\sup_{u\ne0}\|u\|_\infty/\|u\|$.
Prefix quotas bound the score error, including a partial update,
by $(|\operatorname{supp}(x^\star)|+1)L_A$. Since $Q_h$ is
$1/\mu$-Lipschitz from the dual to the primal norm, the resulting
coordinate error is at most
$c_{\|\cdot\|}\eta_T(|\operatorname{supp}(x^\star)|+1)L_A/\mu$.
To leave at least $\delta$ probability on each response, the theorem
requires
\begin{equation}
\label{eq:app-pure-step-size}
    0<\eta_T
    \le
    \frac{\mu}{
        c_{\|\cdot\|}
        \bigl(|\operatorname{supp}(x^\star)|+1\bigr)
        L_A
    }
    \left(\frac1m-\delta\right).
\end{equation}
Lemma~\ref{lem:app-shadowing} proves this control along the entire
update segment. For fixed game data and $\delta$, the condition
holds eventually whenever $\eta_T\to0$.

\paragraph{Prefix-balanced pure implementation.}
For pure trajectories, the cumulative mass $X_t(\ell)$
from Section~\ref{sec:construction} is the integer number of
plays of row $\ell$ before round $t$.
We seek a pure trajectory whose counts
track $tx^\star$ at every prefix. We
also require the randomized construction to
preserve the marginal action probabilities, so
that its terminal discrepancies have mean
zero. Cumulative rounding supplies both properties.
\begin{lemma}[Cumulative-rounding implementation]
\label{lem:app-cumulative-rounding}
For every $x^\star\in\Delta(\mathcal{A})$ and every horizon $T$,
there is a polynomial-time randomized construction
of a pure trajectory $\pi$ with
$\pi_t=e_{\ell_t}$ such that every realization satisfies
\[
    \lfloor t x_\ell^\star\rfloor
    \le X_t(\ell)\le
    \lceil t x_\ell^\star\rceil
    \qquad
    \text{for every }\ell\in\mathcal{A},\ t\in\{0,\ldots,T\},
\]
and, in addition,
\[
    \mathbb P(\ell_t=\ell)=x_\ell^\star
    \qquad
    \text{for every }\ell\in\mathcal{A},\ t\in\{0,\ldots,T-1\}.
\]
Consequently, every realization satisfies
\[
    \sum_{\ell\in\operatorname{supp}(x^\star)}
    \bigl|X_t(\ell)-t x_\ell^\star\bigr|
    \le |\operatorname{supp}(x^\star)|
    \qquad
    \text{for every }t\le T,
\]
and $\mathbb E[X_T(\ell)-Tx_\ell^\star]=0$ for every $\ell\in\mathcal{A}$.
\end{lemma}

\begin{proof}
We encode each round by a star:
selecting one edge chooses the pure action
for that round, while cumulative edge counts
record the prefix action counts. Apply cumulative
rounding \cite[Theorem~4]{golz2025apportionment} to these stars over
the $T$ time steps. Each star
has a center $c$ and a
leaf $v_\ell$ for each optimizer action
$\ell\in\mathcal{A}$. Give the edge $g_\ell=(c,v_\ell)$ fractional
weight $x_\ell^\star$ at every time step.
\begin{center}
\begin{tikzpicture}[
    centervertex/.style={circle,draw,fill=gray!15,minimum size=9mm,inner sep=0pt},
    actionvertex/.style={circle,draw,fill=blue!8,minimum size=9mm,inner sep=0pt},
    edgelabel/.style={fill=white,inner sep=1pt,font=\small}
]
    \node[font=\small] at (0,2.05) {one action per round};
    \node[font=\small] at (4.5,2.05)
        {one leaf per optimizer action};

    \node[centervertex] (c) at (0,0) {$c$};
    \node[actionvertex] (v1) at (4.5,1.3) {$v_1$};
    \node[actionvertex] (v2) at (4.5,0.25) {$v_2$};
    \node at (4.5,-0.5) {$\vdots$};
    \node[actionvertex] (vn) at (4.5,-1.3) {$v_{n}$};

    \draw (c) -- node[edgelabel,above,sloped] {$x_1^\star$} (v1);
    \draw (c) -- node[edgelabel,above,sloped] {$x_2^\star$} (v2);
    \draw (c) -- node[edgelabel,above,sloped] {$x_{n}^\star$} (vn);
\end{tikzpicture}

\smallskip \small At each time step, selecting
edge $g_\ell$ means playing pure action $e_\ell$.
\end{center}
The center has fractional degree one, so
degree preservation selects exactly one action
in each round. Marginal preservation gives
$\mathbb P(\ell_t=\ell)=x_\ell^\star$, and cumulative degree preservation
at each leaf gives the floor--ceiling
bound on $X_t(\ell)$ for every prefix.

If $x_\ell^\star=0$, the quota bound forces $X_t(\ell)=0$,
so only rows in $\operatorname{supp}(x^\star)$ can appear.
Each remaining discrepancy has absolute value less
than one, which gives the stated $\ell_1$
bound. Finally, $\mathbb E[X_T(\ell)]=Tx_\ell^\star$ by marginal preservation.
\end{proof}

For the deterministic guarantee, we need to
choose among the feasible rounding outcomes. The
following lemma shows that we can preserve
any linear function of the terminal
counts while keeping all prefix quotas.
Later, this function will combine the
terminal-potential correction and the active bonus.

\begin{lemma}[Linear-cost derandomization]
\label{lem:app-cost-derandomization}
Fix arbitrary real coefficients $(c_\ell)_{\ell\in\mathcal{A}}$. There is
a deterministic cumulative-rounding realization satisfying the same
prefix quota bounds as Lemma~\ref{lem:app-cumulative-rounding} and
\[
    \sum_{\ell\in\mathcal{A}} X_T(\ell)c_\ell
    \ge
    T\sum_{\ell\in\mathcal{A}} x_\ell^\star c_\ell.
\]
When $x^\star$ and the coefficients are rational,
such a realization can be found in
time polynomial in their encoding length and
in $T,n$.
\end{lemma}

\begin{proof}
At each rounding step, the current fractional
state is a convex combination of two
feasible successors. A linear objective cannot decrease
in both directions, so we can choose
a successor that preserves its value.

Start from the fractional trajectory that plays
$x^\star$ in every round. To guide the
rounding choices, we track its linear value
through
\[
C = \sum_{t=0}^{T-1}\sum_{\ell\in\mathcal A} p_{t,\ell}c_\ell,
\]
where $p_{t,\ell}$ is the current fractional mass
assigned to action $\ell$ in round $t$.
Initially, $p_{t,\ell}=x^\star_\ell$ for every $t,\ell$, so
\[
C_{\mathrm{initial}}
= \sum_{t=0}^{T-1}\sum_{\ell\in\mathcal A}x^\star_\ell c_\ell
= T\sum_{\ell\in\mathcal A}x^\star_\ell c_\ell.
\]

Apply the cumulative-rounding procedure from Lemma~\ref{lem:app-cumulative-rounding}.
An intermediate state encodes the entire
trajectory, with some choices integral and
others still fractional. At a nonintegral
state $P$, a pipage-rounding step produces
feasible successors $P^+$ and $P^-$ satisfying
\[
P = \lambda P^+ + (1-\lambda)P^-,
\qquad \lambda\in[0,1].
\]
Both successors preserve the cumulative-rounding constraints, including
the prefix quotas. We use this feasibility
property of the pipage-rounding construction of \citet{golz2025apportionment}.

Because $C$ is linear in the action-assignment
variables,
\[
C(P) = \lambda C(P^+) + (1-\lambda)C(P^-).
\]
At least one successor therefore satisfies $C(P^\pm)\ge
C(P)$. Choose such a successor deterministically and
continue. At termination, the trajectory is integral
and selects exactly one action in every
round:
\[
\pi_t=e_{\ell_t}, \qquad t=0,\ldots,T-1.
\]
Feasibility at every step gives the required
prefix quotas:
\[
\lfloor t x^\star_\ell\rfloor
\le X_t(\ell)
\le \lceil t x^\star_\ell\rceil
\qquad \text{for every }t,\ell.
\]

The linear value never decreases, so
\[
C_{\mathrm{final}}
\ge C_{\mathrm{initial}}
= T\sum_{\ell\in\mathcal A}x^\star_\ell c_\ell.
\]
In the final trajectory, action $\ell$ appears
$X_T(\ell)$ times, giving
\[
C_{\mathrm{final}}
= \sum_{\ell\in\mathcal A}X_T(\ell)c_\ell.
\]
Therefore,
\[
\sum_{\ell\in\mathcal A}X_T(\ell)c_\ell
\ge
T\sum_{\ell\in\mathcal A}x^\star_\ell c_\ell.
\]

For rational $x^\star$ and $c_\ell$, the
cited rounding construction runs in polynomial
time. Each deterministic choice requires only
comparing two rational linear values, so
the same complexity bound holds here.
\end{proof}

\paragraph{Uniform shadowing of the equilibrium ray.}
We now translate prefix quotas into control
of the learner's probabilities. The one-step Bregman
bound will use the full score segment
between successive rounds, so we need this
control along each segment as well as
at its endpoints. For a sequence supplied
by Lemma~\ref{lem:app-cumulative-rounding}, we follow the response
along this score segment by writing
\[
    y_t(r)
    :=
    Q_h\bigl(\eta_T S_t-r\eta_T A^{\mathsf T} e_{\ell_t}\bigr),
    \qquad r\in[0,1].
\]
Thus $y_t(0)=y_t$ and $y_t(1)=y_{t+1}$. The next lemma
keeps every best-response coordinate in the interval
where the curvature bound $M$ applies.

\begin{lemma}[Shadowing and interiority]
\label{lem:app-shadowing}
Fix $\delta\in(0,1/m)$ and suppose \eqref{eq:app-pure-step-size} holds. Then,
for every realization of the cumulative-rounding sequence,
every round $t$, every $r\in[0,1]$, and every
$p\in\operatorname{br}(x^\star)$,
\[
    y_{t,p}(r)\in[\delta,1-\delta].
\]
\end{lemma}

\begin{proof}
Along the fixed max-min score ray, each
best response has probability at least $1/m$.
Prefix quotas bound the distance from this
ray, and Lipschitz continuity of $Q_h$ transfers
that bound to the probabilities. The step-size
condition ensures that this perturbation leaves at
least $\delta$ mass on each best response.

The score of the pure sequence is
\[
    S_t
    =
    -\sum_{\ell\in\operatorname{supp}(x^\star)}
    X_t(\ell)A^{\mathsf T} e_\ell.
\]
Hence the interpolated score differs from the
max-min ray by
\begin{align*}
    &\left\|
        \eta_T S_t-r\eta_T A^{\mathsf T} e_{\ell_t}
        +\eta_T tA^{\mathsf T} x^\star
    \right\|_* \\
    &\le
    \eta_T
    \sum_{\ell\in\operatorname{supp}(x^\star)}
    \bigl|X_t(\ell)-t x_\ell^\star\bigr|
    \|A^{\mathsf T} e_\ell\|_*
    +r\eta_T\|A^{\mathsf T} e_{\ell_t}\|_* \\
    &\le
    \eta_T\bigl(|\operatorname{supp}(x^\star)|+1\bigr)
    \max_{\ell\in\operatorname{supp}(x^\star)}
    \|A^{\mathsf T} e_\ell\|_*,
\end{align*}
where we used Lemma~\ref{lem:app-cumulative-rounding} and the fact
that $\ell_t\in\operatorname{supp}(x^\star)$.

To compare probabilities, the pointwise \ac{FTRL} bound
of \citet[Lemma~3.6]{su2026nocoin} gives
\[
    \left[Q_h\bigl(-\eta_T tA^{\mathsf T} x^\star\bigr)\right]_p
    \ge\frac1m
    \qquad
    \text{for every }p\in\operatorname{br}(x^\star).
\]
Since $Q_h$ is $1/\mu$-Lipschitz,
\[
    \left\|
        y_t(r)-Q_h\bigl(-\eta_T tA^{\mathsf T} x^\star\bigr)
    \right\|
    \le
    \frac{\eta_T\bigl(|\operatorname{supp}(x^\star)|+1\bigr)}{\mu}
    \max_{\ell\in\operatorname{supp}(x^\star)}
    \|A^{\mathsf T} e_\ell\|_*.
\]
Converting this norm bound to a coordinatewise
bound gives
\begin{align*}
    y_{t,p}(r)
    &\ge
    \frac1m
    -
    c_{\|\cdot\|}
    \frac{\eta_T\bigl(|\operatorname{supp}(x^\star)|+1\bigr)}{\mu}
    \max_{\ell\in\operatorname{supp}(x^\star)}
    \|A^{\mathsf T} e_\ell\|_* \\
    &\ge\delta,
\end{align*}
where the last inequality is \eqref{eq:app-pure-step-size}. Because
$\operatorname{br}(x^\star)$ contains the distinct indices $i$ and
$j$, every $p\in\operatorname{br}(x^\star)$ has another best-response coordinate
whose mass is at least $\delta$. Thus
$y_{t,p}(r)\le1-\delta$.
\end{proof}

\paragraph{Payoff drop and surplus along one update.}
We now relate the payoff-drop calculation of
\citet[Lemma~F.8]{su2026nocoin} to the Bregman terms in our reward
decomposition. For an arbitrary optimizer action $x_t$, write
$v_t:=A^{\mathsf T}x_t$ and interpolate the learner's response as
$y_t(r):=Q_h(\eta_T(S_t-rv_t))$, $0\le r\le1$.
The optimizer receives payment at $y_t(0)$, before the learner
absorbs the update. The following estimate quantifies this benefit.

\begin{proposition}[Payoff drop and Bregman surplus in one update]
\label{prop:app-active-ftrl-baseline}
Suppose two distinct learner actions $i,j$ satisfy
$y_{t,i}(r),y_{t,j}(r)\in[\delta,1-\delta]$ throughout the update,
where $\delta\in(0,1/m)$. Set
$M:=\sup_{u\in[\delta,1-\delta]}\theta''(u)<\infty$. Then
\begin{equation}
\label{eq:app-ftrl-payoff-drop}
\begin{aligned}
  \langle v_t,y_t\rangle-\langle v_t,y_{t+1}\rangle
  &\ge \eta_T(v_{t,i}-v_{t,j})^2
      \int_0^1\frac{dr}
      {\theta''(y_{t,i}(r))+\theta''(y_{t,j}(r))}\\
  &\ge \frac{\eta_T}{2M}(v_{t,i}-v_{t,j})^2.
\end{aligned}
\end{equation}
Moreover, the contribution of this update to the reward decomposition
\eqref{eq:app-disc-cont-used} satisfies
\begin{equation}
\label{eq:app-ftrl-bregman-bonus}
  \frac1{\eta_T}D_{h^*}(\eta_TS_{t+1},\eta_TS_t)
  \ge \frac{\eta_T}{4M}(v_{t,i}-v_{t,j})^2.
\end{equation}
\end{proposition}

\paragraph{One pure step creates a Bregman bonus.}
Lemma~\ref{lem:app-shadowing} ensures that the two distinguished
responses retain mass throughout every pure update. Applying
Proposition~\ref{prop:app-active-ftrl-baseline} therefore gives the
following bound in terms of the chosen row's payoff difference.

\begin{lemma}[One-step pure-action bonus]
\label{lem:app-one-step-bonus}
Under the assumptions of Lemma~\ref{lem:app-shadowing}, every round
$t$ satisfies
\[
    \frac{1}{\eta_T}
    D_{h^*}\bigl(\eta_T S_{t+1},\eta_T S_t\bigr)
    \ge
    \frac{\eta_T}{4M}
    \bigl(A_{\ell_t i}-A_{\ell_t j}\bigr)^2.
\]
\end{lemma}

\begin{proofof}{Proposition~\ref{prop:app-active-ftrl-baseline} and Lemma~\ref{lem:app-one-step-bonus}.}
We prove the proposition for an arbitrary update $v_t$; the pure
case then follows from Lemma~\ref{lem:app-shadowing}. We first bound
the payoff drop along the update and then integrate it to obtain
the Bregman bonus. Following \citet[Lemma~F.8]{su2026nocoin}, write
\[
    v':=v_t,
    \qquad
    y(r):=y_t(r).
\]
Although that lemma is stated for the
alternating trajectory used there, its calculation only
requires an affine score path of the
form $z-r\eta_T v'$. Its proof therefore applies
verbatim here with $z=\eta_T S_t$. The identity
expresses the rate of payoff change as
a variance weighted by inverse curvature. For
$r\in[0,1]$, we collect the inverse-curvature weights, their
total mass, and their normalized distribution as
\[
    w_r(q)
    :=
    \begin{cases}
        1/\theta''(y_q(r)),&q\in\operatorname{supp}(y(r)),\\
        0,&q\notin\operatorname{supp}(y(r)),
    \end{cases}
    \qquad
    W_r:=\sum_qw_r(q),
    \qquad
    \nu_r(q):=\frac{w_r(q)}{W_r}.
\]
To represent this weighted variance probabilistically, let
$I_r\sim\nu_r$ and $Z_r:=v'_{I_r}$. The variance identity and
its two-coordinate lower bound give, for almost
every $r\in[0,1]$,
\[
    -\frac{d}{dr}\langle v',y(r)\rangle
    =
    \eta_T W_r\operatorname{Var}(Z_r),
    \qquad
    W_r\operatorname{Var}(Z_r)
    \ge
    \frac{(v'_i-v'_j)^2}
    {\theta''(y_i(r))+\theta''(y_j(r))}.
\]
By hypothesis, both distinguished coordinates lie in $[\delta,1-\delta]$,
so each of their second derivatives is
at most $M$. The two-coordinate bound therefore
yields
\[
    -\frac{d}{dr}\langle v',y(r)\rangle
    \ge
    \frac{\eta_T(v'_i-v'_j)^2}{2M}
\]
for almost every $r$. Since $\langle v',y(\cdot)\rangle$
is absolutely continuous,
\[
    \langle v',y(0)\rangle-\langle v',y(r)\rangle
    \ge
    r\frac{\eta_T(v_{t,i}-v_{t,j})^2}{2M}.
\]
Integrating the two-coordinate variance bound gives
\eqref{eq:app-ftrl-payoff-drop}. To obtain the Bregman bonus, use
$S_{t+1}=S_t-v'$ and $\nabla h^*=Q_h$:
\begin{align*}
    \frac{1}{\eta_T}
    D_{h^*}\bigl(\eta_T S_{t+1},\eta_T S_t\bigr)
    &=
    \frac{1}{\eta_T}
    \left(
        h^*\bigl(\eta_T S_t-\eta_T v'\bigr)-h^*\bigl(\eta_T S_t\bigr)
        +\eta_T\left\langle Q_h\bigl(\eta_T S_t\bigr),v'\right\rangle
    \right) \\
    &=
    \int_0^1
    \bigl(\langle v',y(0)\rangle-\langle v',y(r)\rangle\bigr)\,dr \\
    &\ge
    \int_0^1
    r\frac{\eta_T(v_{t,i}-v_{t,j})^2}{2M}\,dr \\
    &=
    \frac{\eta_T(v_{t,i}-v_{t,j})^2}{4M}.
\end{align*}
This proves \eqref{eq:app-ftrl-bregman-bonus}. The endpoint payoff
drop has coefficient $1/(2M)$, whereas averaging it along the
update contributes the extra factor $\int_0^1r\,dr=1/2$.
For a pure action, $v_t=A^{\mathsf T}e_{\ell_t}$, and
Lemma~\ref{lem:app-shadowing} supplies the required interiority;
substitution gives Lemma~\ref{lem:app-one-step-bonus}.
\end{proofof}

\paragraph{Combining the terminal potential and the active
bonus.}
We now assemble these ingredients into the
proof of Theorem~\ref{thm:app-pure-active}.

\begin{proofof}{Theorem~\ref{thm:app-pure-active}.}
We return to the reward decomposition in
\eqref{eq:app-disc-cont-used}. Both the terminal potential and the
sum of one-step bonuses differ from their
target values by terms linear in the
terminal count discrepancies. We retain these terms
and combine them into a single
cost, then choose a rounding outcome
for which their total contribution is
nonnegative. The remaining error is the
quadratic remainder in the potential bound.

First compare the terminal potential with that
of the fixed max-min strategy. Because $h^*$
is $1/\mu$-smooth, $\Phi_{\eta_T}$ is $\eta_T/\mu$-smooth, giving
\begin{align*}
    \Phi_{\eta_T}\bigl(S_T\bigr)
    &\le
    \Phi_{\eta_T}\bigl(-TA^{\mathsf T} x^\star\bigr)
    +\left\langle
        Q_h\bigl(-\eta_T TA^{\mathsf T} x^\star\bigr),
        S_T+TA^{\mathsf T} x^\star
    \right\rangle
    +\frac{\eta_T}{2\mu}
    \left\|S_T+TA^{\mathsf T} x^\star\right\|_*^2.
\end{align*}
The terminal instance of the prefix
quota bounds controls the quadratic remainder:
\[
    \left\|S_T+TA^{\mathsf T} x^\star\right\|_*
    =
    \left\|
        -\sum_{\ell\in\mathcal{A}}
        \bigl(X_T(\ell)-Tx_\ell^\star\bigr)A^{\mathsf T} e_\ell
    \right\|_*
    \le
    |\operatorname{supp}(x^\star)|
    \max_{\ell\in\operatorname{supp}(x^\star)}
    \|A^{\mathsf T} e_\ell\|_*.
\]
The fixed-strategy identity in \eqref{eq:app-zero-sum-continuous-reward} gives
\[
    \Phi_{\eta_T}(0)-\Phi_{\eta_T}\bigl(-TA^{\mathsf T} x^\star\bigr)
    =
    R_{\mathrm{cont}}(x^\star,T)
    =
    T\operatorname{Val}(A)+\operatorname{LAG}_{x^\star}(T).
\]
Using this identity and the score
bound in the potential comparison gives
\begin{equation}
\label{eq:app-terminal-bound}
\begin{aligned}
    \Phi_{\eta_T}(0)-\Phi_{\eta_T}\bigl(S_T\bigr)
    &\ge
    T\operatorname{Val}(A)
    +\operatorname{LAG}_{x^\star}(T) \\
    &\quad-
    \left\langle
        Q_h\bigl(-\eta_T TA^{\mathsf T} x^\star\bigr),
        S_T+TA^{\mathsf T} x^\star
    \right\rangle \\
    &\quad-
    \frac{\eta_T|\operatorname{supp}(x^\star)|^2}{2\mu}
    \left(
        \max_{\ell\in\operatorname{supp}(x^\star)}
        \|A^{\mathsf T} e_\ell\|_*
    \right)^2.
\end{aligned}
\end{equation}

We next control the active contribution. Lemma~\ref{lem:app-one-step-bonus}
lower-bounds each round's Bregman term by
the squared payoff difference of the
chosen row. Summing and separating the
target counts from their discrepancies gives
\begin{align}
    \frac{1}{\eta_T}
    \sum_{t=0}^{T-1}
    D_{h^*}\bigl(\eta_T S_{t+1},\eta_T S_t\bigr)
    &\ge
    \frac{\eta_T}{4M}
    \sum_{t=0}^{T-1}
    (A_{\ell_t i}-A_{\ell_t j})^2
    \notag\\
    &=
    \frac{\eta_T T}{4M}
    \sum_{\ell\in\operatorname{supp}(x^\star)}
    x_\ell^\star(A_{\ell i}-A_{\ell j})^2
    \notag\\
    &\quad+
    \sum_{\ell\in\operatorname{supp}(x^\star)}
    \bigl(X_T(\ell)-Tx_\ell^\star\bigr)
    \frac{\eta_T(A_{\ell i}-A_{\ell j})^2}{4M}.
\label{eq:app-bonus-sum}
\end{align}
Combining \eqref{eq:app-disc-cont-used}, \eqref{eq:app-terminal-bound}, and \eqref{eq:app-bonus-sum}, every cumulative-rounding
realization obeys
\begin{align*}
    R(\pi,T)
    &\ge
    T\operatorname{Val}(A)
    +\operatorname{LAG}_{x^\star}(T) \\
    &\quad+
    \frac{\eta_T T}{4M}
    \sum_{\ell\in\operatorname{supp}(x^\star)}
    x_\ell^\star(A_{\ell i}-A_{\ell j})^2 \\
    &\quad-
    \frac{\eta_T|\operatorname{supp}(x^\star)|^2}{2\mu}
    \left(
        \max_{\ell\in\operatorname{supp}(x^\star)}
        \|A^{\mathsf T} e_\ell\|_*
    \right)^2 \\
    &\quad+
    \sum_{\ell\in\operatorname{supp}(x^\star)}
    \bigl(X_T(\ell)-Tx_\ell^\star\bigr)c_\ell,
\end{align*}
where the coefficient of each terminal discrepancy
collects its effect on both the
terminal potential and the active bonus:
\[
    c_\ell
    :=
    e_\ell^{\mathsf T} A
    Q_h\bigl(-\eta_T TA^{\mathsf T} x^\star\bigr)
    +
    \frac{\eta_T(A_{\ell i}-A_{\ell j})^2}{4M}.
\]
Thus only a linear function of the
terminal discrepancies remains to be controlled. For
the randomized construction, Lemma~\ref{lem:app-cumulative-rounding} gives $\mathbb
E[X_T(\ell)-Tx_\ell^\star]=0$ for every $\ell$, and hence
\[
    \mathbb E\left[
        \sum_{\ell\in\operatorname{supp}(x^\star)}
        \bigl(X_T(\ell)-Tx_\ell^\star\bigr)c_\ell
    \right]=0.
\]
Since $\operatorname{LAG}_{x^\star}(T)\ge0$, discarding this
passive contribution proves \eqref{eq:app-pure-active-main}
in expectation.

For the deterministic claim, apply Lemma~\ref{lem:app-cost-derandomization} to
the coefficients $c_\ell$ above. Its conclusion is
exactly
\[
    \sum_{\ell\in\operatorname{supp}(x^\star)}
    \bigl(X_T(\ell)-Tx_\ell^\star\bigr)c_\ell
    =
    \sum_{\ell\in\mathcal{A}}X_T(\ell)c_\ell
    -T\sum_{\ell\in\mathcal{A}}x_\ell^\star c_\ell
    \ge0.
\]
Substituting this into the preceding pointwise bound and discarding
the nonnegative passive contribution proves
\eqref{eq:app-pure-active-main} for the resulting deterministic pure sequence.
\end{proofof}

Algorithm~\ref{alg:app-active-pure-construction} records the deterministic construction: the linear
costs combine the terminal-potential correction with each
row's contribution to the active bonus.

\begin{algorithm2e}[t]
\DontPrintSemicolon
\caption{Cost-guided cumulative rounding for active extraction}
\label{alg:app-active-pure-construction}
\KwIn{A max-min strategy $x^\star$, horizon $T$, learning
rate $\eta_T$, best responses $i,j$, and curvature
bound $M$.} \KwOut{A deterministic pure trajectory satisfying
the prefix quotas and the active
linear-cost bound.} For every $\ell\in\mathcal{A}$, set
\[
  c_\ell\gets e_\ell^{\mathsf T}A Q_h(-\eta_T T A^{\mathsf T}x^\star)
       +\frac{\eta_T(A_{\ell i}-A_{\ell j})^2}{4M}.
\]
Create $T$ copies of a star with
one leaf for each $\ell\in\operatorname{supp}(x^\star)$, giving its
edge weight $x_\ell^\star$ and linear cost $c_\ell$\;
Apply the cumulative-rounding auxiliary-graph construction of \citet{golz2025apportionment},
assigning zero cost to auxiliary edges\; \While{the
auxiliary graph has a fractional edge}{ Find
a fractional even cycle or maximal fractional
path and form the two feasible pipage
successors by maximal alternating updates\; Select
a successor whose linear objective is
no smaller than the current value\;
}
\For{$t=0,\ldots,T-1$}{ Let $\ell_t$ be the unique leaf
whose original edge in copy $t$ is
rounded to one, and set $\pi_t\gets e_{\ell_t}$\;
}
\Return{$\pi=(\pi_0,\ldots,\pi_{T-1})$}\;
\end{algorithm2e}

Non-identification makes the coefficient of $\eta_T T$
strictly positive. For a sufficiently large horizon,
this term dominates the quadratic rounding loss.
The next corollary records the corresponding threshold.

\begin{corollary}[Asymptotic $\Omega(\eta_T T)$ pure surplus]
\label{cor:app-pure-active-asymptotic}
Under the assumptions of Theorem~\ref{thm:app-pure-active}, if
\[
    T
    \ge
    \frac{
        4M|\operatorname{supp}(x^\star)|^2
        \left(
            \max_{\ell\in\operatorname{supp}(x^\star)}
            \|A^{\mathsf T} e_\ell\|_*
        \right)^2
    }{
        \mu
        \sum_{\ell\in\operatorname{supp}(x^\star)}
        x_\ell^\star(A_{\ell i}-A_{\ell j})^2
    },
\]
then there exists a pure trajectory $\pi$
satisfying
\[
    R(\pi,T)
    \ge
    T\operatorname{Val}(A)
    +
    \frac{\eta_T T}{8M}
    \sum_{\ell\in\operatorname{supp}(x^\star)}
    x_\ell^\star(A_{\ell i}-A_{\ell j})^2.
\]
In particular, for a fixed game and
fixed regularizer geometry, non-identification on the max-min
support implies pure active surplus $\Omega_{A,h}(\eta_T T)$.
\end{corollary}

\begin{proof}
The horizon condition makes the rounding loss
at most half of the active term
in Theorem~\ref{thm:app-pure-active}:
\[
    \frac{\eta_T|\operatorname{supp}(x^\star)|^2}{2\mu}
    \left(
        \max_{\ell\in\operatorname{supp}(x^\star)}
        \|A^{\mathsf T} e_\ell\|_*
    \right)^2
    \le
    \frac{\eta_T T}{8M}
    \sum_{\ell\in\operatorname{supp}(x^\star)}
    x_\ell^\star(A_{\ell i}-A_{\ell j})^2.
\]
Substituting this inequality into Theorem~\ref{thm:app-pure-active}
gives the stated bound.
\end{proof}

\subsection{Polynomial-time construction and identification}
\label{app:algorithmic}

\begin{remark}[Polynomial-time construction and identification]
\label{rem:app-algorithmic-pure}
To implement the construction, we must find
a max-min strategy with a suitable pair
of best responses and approximate the resulting
linear costs. For rational input, both tasks
can be performed to any prescribed additive
accuracy in polynomial time, subject to polynomial-time
evaluation of $Q_h$ to the required precision.

First compute $\operatorname{Val}(A)$ by linear programming. For
each ordered pair $i\ne j$, solve
\[
\begin{aligned}
    \max_{x\in\mathbb R^{\mathcal{A}}}\quad
    &\sum_{\ell\in\mathcal{A}} x_\ell(A_{\ell i}-A_{\ell j})^2 \\
    \text{subject to}\quad
    &x\ge0,\qquad \mathbf 1^{\mathsf T} x=1,\\
    &A^{\mathsf T} x\ge \operatorname{Val}(A)\mathbf 1,\\
    &(A^{\mathsf T} x)_i=(A^{\mathsf T} x)_j=\operatorname{Val}(A).
\end{aligned}
\]
The feasible points are precisely max-min
strategies for which $i$ and $j$
are learner best responses. Thus a
positive optimum supplies the required $(x^\star,i,j)$.

It remains to obtain the costs used
by the rounding procedure. The proof of
Theorem~\ref{thm:app-pure-active} uses exact coefficients $c_\ell$ for its
existential guarantee. For the computational guarantee, evaluate
$Q_h(-\eta_T T A^{\mathsf T} x^\star)$ accurately enough
to obtain rational coefficients $\widetilde c_\ell$ with
$\max_\ell|\widetilde c_\ell-c_\ell|\le\varepsilon$, and apply Lemma~\ref{lem:app-cost-derandomization} to $\widetilde
c_\ell$. The prefix quota bound then gives
\[
    \sum_{\ell\in\operatorname{supp}(x^\star)}
    \bigl(X_T(\ell)-Tx_\ell^\star\bigr)c_\ell
    \ge-|\operatorname{supp}(x^\star)|\varepsilon.
\]
Thus the lower bound in Theorem~\ref{thm:app-pure-active}
loses at most $|\operatorname{supp}(x^\star)|\varepsilon$, and the
resulting pure sequence is computable in
time polynomial in $T,n,m$, the rational
input length, and the requested precision.
\end{remark}

\subsection{Random games}
\label{app:random-games}

The active-surplus theorem requires a pair of
best responses with different payoff columns on
the max-min support. For random games, the
estimates of \citet[Appendix~F]{su2026nocoin} supply this condition and
a quantitative payoff gap. Lemma~F.2 supplies the
non-identification event, and Lemma~F.5 bounds every row-entry
gap below by $2/(n^2m^3)$ with probability at
least $1-1/(nm)$. Applying the deterministic theorem on
their intersection gives the following guarantee.

\begin{corollary}[Pure active extraction in random games]
\label{cor:app-random-pure}
Let the entries of the $n$-by-$m$ matrix
$A$ be independent and uniformly distributed
on $[-1,1]$. Fix $\delta\in(0,1/m)$ and suppose
\[
    \eta_T
    \le
    \frac{\mu}{
        c_{\|\cdot\|}(n+1)
        \sup_{u\in[-1,1]^m}\|u\|_*
    }
    \left(\frac1m-\delta\right).
\]
Then, with probability at least
\[
    1-
    \frac{n!m!}{(n+m-1)!}
    -\frac1{nm},
\]
for every given horizon $T$, there
exists a pure trajectory $\pi$ satisfying
\[
    R(\pi,T)
    \ge
    T\operatorname{Val}(A)
    +
    \frac{\eta_T T}{Mn^4m^6}
    -
    \frac{\eta_T n^2}{2\mu}
    \left(
        \sup_{u\in[-1,1]^m}\|u\|_*
    \right)^2.
\]
Consequently, if
\[
    T
    \ge
    \frac{Mn^6m^6}{\mu}
    \left(
        \sup_{u\in[-1,1]^m}\|u\|_*
    \right)^2,
\]
then
\[
    R(\pi,T)
    \ge
    T\operatorname{Val}(A)
    +
    \frac{\eta_T T}{2Mn^4m^6}.
\]
\end{corollary}

\begin{proof}
We use the two cited events to
bound the active term from below, then
bound the rounding loss by the total
number of rows. A union bound gives
the stated probability for their intersection. On
this intersection, Lemmas~F.2 and~F.5 of \cite{su2026nocoin} provide
a max-min strategy $x^\star$ and distinct
$i,j\in\operatorname{br}(x^\star)$, and every row $\ell$ satisfies
\[
    |A_{\ell i}-A_{\ell j}|
    \ge
    \frac{2}{n^2m^3}.
\]
Therefore
\[
    \sum_{\ell\in\operatorname{supp}(x^\star)}
    x_\ell^\star(A_{\ell i}-A_{\ell j})^2
    \ge
    \frac{4}{n^4m^6}.
\]
Also, $|\operatorname{supp}(x^\star)|\le n$ and
\[
    \max_{\ell\in\operatorname{supp}(x^\star)}
    \|A^{\mathsf T} e_\ell\|_*
    \le
    \sup_{u\in[-1,1]^m}\|u\|_*.
\]
These estimates also ensure that the assumed
step-size bound implies \eqref{eq:app-pure-step-size}. Substituting them
into Theorem~\ref{thm:app-pure-active} gives
the first reward bound. The stated
horizon condition then allows us to
apply Corollary~\ref{cor:app-pure-active-asymptotic}, giving the second bound.
\end{proof}

% !TEX root = ../main.tex
\section{Extended related work}
\label{app:extended-related-work}

Strategizing against a learner sits between equilibrium
computation, online learning, and optimal control. To
place our results within this landscape, it
helps to keep three choices in view:
what the optimizer knows, which learner is
fixed, and which reward benchmark the
optimizer seeks. Changing any of these
choices can change the computational problem.

\subsection*{Static commitment and dynamic control}

We begin with the classical commitment benchmark.
For an optimizer strategy $x$, let $\operatorname{br}_B(x):=\operatorname*{arg\,max}_{b\in\mathcal{B}}x^{\mathsf
T}Be_b$ be the learner's pure best responses.
With ties resolved in the optimizer's
favor, the one-shot Stackelberg value is
\begin{equation*}
\operatorname{SV}(A,B) :=\max_{x\in\Delta(\mathcal{A})} \max_{b\in\operatorname{br}_B(x)}x^{\mathsf T}Ae_b.
\end{equation*}
For explicitly given two-player payoff matrices, this
value is computable in polynomial time \citep{conitzer2006commitment}.
One considers each follower action $b$ in
turn and optimizes the leader's payoff over
the linear constraints making $b$ a best
response. The favorable tie convention is part
of this benchmark; it does not assert
that a prescribed learning algorithm selects the
optimizer's preferred response at a tie.

The repeated-game problem adds a state to
this static picture. An optimizer action changes
the learner's accumulated experience, which changes later
responses and rewards. Consequently, a profitable trajectory
may deliberately visit states that repeated play
of a commitment strategy would never reach.
\citet{deng2019strategizing} show that mean-based learners can allow
a linear improvement over the Stackelberg benchmark,
whereas no-swap-regret learners cap the optimizer's reward
at that benchmark up to a sublinear
term. Their control formulation exposes the planning
problem, without supplying a general efficient optimizer
for it. Our hardness concerns this full
trajectory objective with a prescribed learner.

The economic literature supplies important structured examples.
\citet{braverman2017sellingnoregretbuyer} show how a seller can extract
almost the entire welfare from a mean-based
buyer, and \citet{guruganesh2024contracting} study utility-maximizing contracts against
a learning agent. \citet{lin2025principalagent} connect repeated principal--agent
interaction to approximate best responses, covering Stackelberg
games, contracting, and persuasion. These settings explain
why exploiting a learner can be valuable;
their structure also matters for the available
algorithms. In Bayesian games, \citet{mansour2022bayesian} study the
corresponding strategic benchmarks and introduce polytope swap
regret, showing that private types require
a more refined treatment of deviations.

\subsection*{Known dynamics and computational hardness}

The closest computational comparisons are the two
papers of Assos and coauthors. \citet{assos2024maximizing} explicitly
anticipate the learner's update rule. In zero-sum
games, they obtain an optimal continuous-time strategy
against replicator dynamics and discrete-time guarantees for
MWU. In general-sum games, they establish hardness
against fictitious play, which need not have
vanishing external regret. This first obstruction concerns
constant additive accuracy, leaving open whether computational
difficulty persists at the scale of the
total reward against standard no-regret dynamics.

\citet{assos2025intractability} answer that question for Hedge/MWU. They
prove an $\Omega(T)$ approximation barrier even when
the game dimensions are bounded by arbitrarily
small fixed powers of the horizon and
the updates are slow. Thus, efficient simulation
of a familiar learning rule does not
imply efficient optimization of the sequence that
drives it. Theorem~\ref{thm:main} broadens the learner class
through the mean-based condition and allows explicitly represented
mixed outputs while comparing their reward with
the best pure trajectory. The extraction argument
must therefore accommodate fractional interleavings, rather than
only schedules selecting one action per round.

Several positive results help locate the boundary
of this obstruction. \citet{guo2023hedge} analyze optimal play
against Hedge in finitely repeated $2\times2$ zero-sum
games, exploiting the low-dimensional state structure. \citet{brown2024control}
give algorithms for locally controllable nonlinear systems
with convex state losses; their repeated-game application
concerns projected-gradient learners and optimizer rewards that
are approximately functions of the learner's state
alone. Such controllability and reward assumptions provide
a route to tractability. They do
not give an unrestricted algorithm for
the arbitrary trajectory benchmark used here.
A useful comparison therefore asks both
how the learner evolves and which
control objective an algorithm actually optimizes.

\subsection*{Unknown payoffs and local Stackelberg targets}

When the learner's payoffs are unknown, even
attaining the static commitment benchmark becomes an
information problem. A best-response oracle can answer
independent queries, but a learner carries its
history from one interaction to the next.
Moving between two query strategies may require
many rounds simply to change the historical
average. \citet{brown2023learners} exhibit games where learning the
Stackelberg strategy is polynomial against no-swap-regret learners
but requires exponentially many rounds against mean-based
no-regret learners. This barrier concerns access
through interaction, whereas our optimizer knows
the game and the learning rule.

\citet{ananthakrishnan2025local} address this difficulty by seeking approximate
local Stackelberg strategies. Their algorithm makes small
changes to historical play and explores adjacent
best-response regions. Under their geometric assumptions and
smoothed analysis, it gives a polynomial-time approximation
scheme: the complexity is polynomial in game
size for fixed accuracy, but exponential in
inverse accuracy. They also prove an exponential
interaction lower bound in inverse accuracy
against fictitious play with unknown learner
payoffs. It should therefore be kept
distinct from hardness of local optimization
with full knowledge of the game.

Locality also has different meanings across models.
\citet{fiez2020implicit} study differential Stackelberg equilibria in continuous
games and derive learning dynamics through the
follower's local response, with convergence guarantees under
the relevant stability conditions. A differential equilibrium,
a local optimum of a piecewise-linear commitment
objective, and a local optimum of a
finite-horizon trajectory are different targets. In particular,
with known two-player matrices and favorable
ties, a global static Stackelberg optimum
is already efficiently computable and is
also locally optimal. The extension suggested
by our hardness is to investigate
locality in the controlled trajectory itself.

Learning to steer an unknown opponent creates
further distinctions. \citet{zhang2025steer} recover information about the
learner's payoffs under additional assumptions on its
dynamics, then use it to steer play.
In a complementary direction, \citet{arunachaleswaran2025unknown} design the
learning algorithm itself to maximize the learner's
utility against a distribution of strategic optimizer
types. Their menu-based formulation treats commitment in
algorithm space. Our computational problem instead fixes
the learner in advance and asks how
well an informed optimizer can respond.

\subsection*{Regret guarantees and learner design}

Regret restrictions constrain attainable rewards, but they
do not generally identify which rewards a
fixed update rule can realize. \citet{brown2023learners} use
generalized equilibria to characterize payoff regions under
asymmetric regret constraints, including the maximal reward
attainable against some no-regret learner. Optimizing over
such a region allows the learner to
vary within its regret class. Computing the
best trajectory against a given Hedge or
\ac{FTRL} rule retains the additional dynamical constraint.

\citet{cai2025proximal} introduce proximal regret and proximal correlated
equilibria. Under their comparator assumptions, projected gradient
descent minimizes this stronger regret notion, and
mirror descent minimizes its Bregman counterpart. These
results concern efficient learner updates and regret
guarantees. To connect them to our problem,
one must also ask whether an optimizer
can efficiently plan against those updates. Combining
their Appendix~C with Theorem~\ref{thm:main} shows that, for
the entropy-based Bregman formulation and admissible comparators,
MWU can have the corresponding regret
guarantee while its optimizer still faces
our approximation barrier. We discuss the
sharper Euclidean projected-gradient question in Appendix~\ref{app:open-questions}.

\subsection*{Zero-sum surplus and long-run objectives}

In zero-sum games, both trajectory benchmarks lie between \(T\operatorname{Val}(A)\) and \(T\operatorname{Val}(A)+\operatorname{Regret}_{\mathsf L}(T)\), so their surplus is sublinear against a no-regret learner. The finer question is how
much of that term can be realized
in a given game. Building on the
entropy-specific analysis of \citet{assos2024maximizing}, \citet{su2026nocoin} develop a
geometric account for separable \ac{FTRL}: fixed max-min
play collects surplus while suboptimal learner actions
disappear, and discrete updates create an additional
Bregman contribution. Their bounds distinguish inverse-rate
passive gains from active gains accumulating
over time, and show how boundary
behavior of the regularizer affects extraction.

Our zero-sum results address the role of
pure control in these mechanisms. Terminal rounding
preserves the passive contribution up to a
support-size loss; prefix balancing also retains
the active contribution under the stated
regularity and step-size conditions. This complements
the general-sum obstruction: the two regimes
differ both in computational difficulty and
in the role of mixed actions.

Finally, finite-horizon reward and long-run average reward
pose different questions. \citet{masoumian2026average} study average-reward reinforcement
learning with model selection and apply it
to opponents with unknown bounded memory. The
induced state space, its memory dependence, and
the required MDP assumptions are central to
their guarantees. Cumulative-score learners can have unbounded
state even in a fixed finite game.
Moreover, an $\Omega(T)$ finite-horizon approximation gap does
not automatically give hardness for an infinite-horizon
value: one must fix a single learning
rule and ensure that the reward
gap survives in the long-run average.

% !TEX root = ../main.tex
\section{Open questions and future work}
\label{app:open-questions}

The results suggest four directions in which
the benchmark, the learner, the time scale,
or the input model changes. Each requires
a separate complexity analysis. Section~\ref{sec:discussion} highlights
directions~(B) and~(D).

\paragraph{(A) Local optima of controlled trajectories.}
The local Stackelberg results of \citet{ananthakrishnan2025local} motivate
an extension from static commitments to dynamic
control. With known matrices and favorable ties,
a global static Stackelberg optimum is already
computable by linear programming. A more demanding
target is local optimality of the full
reward $R(x,T)$. For a prescribed local radius
$\rho\in(0,1)$, can we efficiently find a mixed
trajectory $x$ such that $R(x',T)\le R(x,T)+\varepsilon T$
whenever $\max_{t<T}\|x'_t-x_t\|_1\le\rho$? These deviations change every round
only slightly, but their effects can accumulate
through the learner's state. Does a polynomial-time
approximation scheme exist for fixed $\varepsilon$
and $\rho$, or can this local
target retain computational hardness? For smooth
learners, first-order stationarity is a natural
starting point, but it does not
by itself certify this finite-neighborhood guarantee.

\paragraph{(B) Strategizing against proximal-regret learners.}
\citet{cai2025proximal} already give efficient algorithms for minimizing
proximal regret. The open computational target is
the optimizer's planning problem against a fixed
learner with such a guarantee. The entropy-based
Bregman version includes MWU under suitable comparator
assumptions, so it already intersects our hard
class. A sharper test case is recursive
Euclidean projected gradient descent, distinct from the
Euclidean \acs{FTRL} covered by our hardness result:
does its proximal-regret guarantee
permit efficient near-optimal steering in arbitrary general-sum
games, beyond the controllable instances studied by
\citet{brown2024control}? A related value question is whether
the Euclidean rule caps reward at $T\operatorname{SV}(A,B)+o(T)$
or can allow a linear excess. The achievable
reward and the complexity of finding a
trajectory attaining it should be examined separately;
efficient regret minimization settles neither question.

\paragraph{(C) Hardness beyond a prescribed horizon.}
Fix one game, one initialization, and one
horizon-independent learner $\mathsf L$, such as Hedge
with a prescribed anytime step-size schedule. Writing
$R_{\mathsf L}$ for reward under this rule,
consider
\begingroup
\setlength{\abovedisplayskip}{3pt}
\setlength{\belowdisplayskip}{3pt}
\setlength{\abovedisplayshortskip}{0pt}
\setlength{\belowdisplayshortskip}{3pt}
\begin{equation*}
  V_\infty(A,B;\mathsf L)
  :=\sup_{x\in\Delta(\mathcal{A})^{\mathbb N}}
    \liminf_{T\to\infty}\frac{R_{\mathsf L}(x_{<T},T)}{T}.
\end{equation*}
\endgroup
Can this value be approximated from the
game encoding and $\varepsilon$, without an explicit
horizon parameter? When do polynomial-size periodic or
finite-memory controls suffice for $\varepsilon$-optimality? Our reduction
does not settle this: its learner parameters
may depend on $T$, and a profitable
finite transient may vanish in the average.
For an anytime no-external-regret learner in a
zero-sum game, the value is already $\operatorname{Val}(A)$
by the regret sandwich. General-sum cumulative-score dynamics
are therefore the main target, with bounded-memory
average-reward methods \citep{masoumian2026average} providing a useful comparison.

% \Needspace{22\baselineskip}
\paragraph{(D) Smoothed complexity.}
Theorem~\ref{thm:main} is a worst-case statement over exactly
specified payoffs. The score identities~\eqref{eq:score-identities} hold with
equality, and the local and edge alarms
are tuned so that the near-leading constraints~\eqref{eq:local-comparisons}
hold within the mean-based slack $\gamma_TT$. Following
the smoothed-analysis framework of \citet{spielman2004smoothed}, consider adding
independent uniform noise in $[-\tau,\tau]$ to every
entry of $A$ and $B$. Along any
fixed optimizer trajectory, this shifts each cumulative
score by at most $\tau T$.
The first question is therefore the largest
$\tau$, as a function of $\gamma_T$ and
the game size, for which the extraction
argument of Section~\ref{sec:soundness} and the $\Omega(T)$ reward
gap survive. For larger $\tau$, do perturbed
instances admit a polynomial-time approximation scheme? The
smoothed guarantees of \citet{ananthakrishnan2025local} for local Stackelberg
targets, under their geometric assumptions and Gaussian
perturbations, provide a comparison for this different
control objective.

The zero-sum side poses a complementary question.
Here we perturb $A$ and set $B=-A$
to preserve the zero-sum structure. Corollary~\ref{cor:app-random-pure} treats
independent uniform entries, corresponding to the zero-centered
case with $\tau=1$. A strict pure saddle
point can survive all sufficiently small bounded
perturbations, so active surplus cannot be guaranteed
uniformly over arbitrary base games. Which base
games and noise scales make Assumption~\ref{ass:app-non-identification} and
a quantitative lower bound on $\sigma_{ij}^2(x^\star)$ hold
with high probability? Such conditions would extend
the active-surplus guarantee of Theorem~\ref{thm:app-pure-active} to smoothed
instances. Payoffs can be rescaled after perturbation
to retain the theorem's normalization.

\end{document}